\documentclass[12pt]{article}

\usepackage{amsmath}
\usepackage{amssymb}
\usepackage{mathtools}
\usepackage{amsthm}
\usepackage{natbib}

\usepackage{graphicx}
\usepackage{epstopdf}

\newtheorem{lemma}{Lemma}
\newtheorem{theorem}{Theorem}
\newtheorem{corollary}{Corollary}

\newcommand{\qsd}{quasi-stationary distribution}

\newcommand{\OU}{Ornstein--Uhlenbeck}
\newcommand{\ds}{\displaystyle}

\newcommand\Z{\ensuremath{\mathbb{Z}}}
\newcommand\R{\ensuremath{\mathbb{R}}}

\newcommand\E{\ensuremath{\mathbb{E}}}

\newcommand{\bl}{\boldsymbol{l}}
\newcommand{\by}{\boldsymbol{y}}

\newcommand{\bn}{\boldsymbol{n}}

\newcommand{\bu}{\boldsymbol{u}}

\newcommand{\bz}{\boldsymbol{z}}
\newcommand{\bX}{\boldsymbol{X}}
\newcommand{\bZ}{\boldsymbol{Z}}
\newcommand{\bXN}{\boldsymbol{X}^{(N)}}
\newcommand{\bZN}{\boldsymbol{Z}^{(N)}}

\newcommand{\myprime}{\,\prime}

\newcommand{\Var}{\text{Var}}
\newcommand{\Cov}{\text{Cov}}
\newcommand{\sgn}{\text{sgn}}
\newcommand{\tr}{\text{tr}}
\newcommand{\norm}[1]{\left\lVert #1 \right\rVert}

\newcommand{\M}{\mathcal{M}}
\newcommand{\bzero}{\boldsymbol{0}}

\newcommand{\bx}{\boldsymbol{x}}

\title{Persistence, Thresholds, and Trait Composition in a Regulated
Mutation--Selection Model}

\author{Phil Pollett}

\begin{document}
\maketitle

\begin{center}
\begin{minipage}{14cm}
{\bf Abstract} \ \ 
{\setlength{\parindent}{1.5em}
We study a population model in which individuals carry one of two traits
and evolve under mutation, selection, and density-dependent regulation.
A deterministic large-population limit yields a nonlinear system
coupling logistic growth with mutation–selection dynamics. We identify
threshold conditions governing extinction, persistence, and long-term
trait composition. In particular, mutation induces an effective
mortality rate that determines whether the population can be sustained.
When inheritance dominates mutation, a second threshold emerges:
population establishment depends on initial trait composition as well as
overall growth rates. Although extinction ultimately occurs, the system
typically exhibits long-lived quasi-equilibrium behaviour. A diffusion
approximation provides a tractable description of this, and
reveals a transition in the sign of trait correlations. The model thus
illustrates how mutation, selection, and resource limitation jointly
shape both ecological persistence and evolutionary outcomes.
}
\end{minipage}
\end{center}



\section{Introduction}

Understanding how mutation and selection interact in populations subject
to resource limitation is a central problem in theoretical population
biology. Classical models either describe changes in trait frequencies
under a fixed population size or analyse population growth without
explicit evolutionary structure. In many biological systems, however,
these processes interact: population size influences selection,
while mutation and selection shape the demographic trajectory. Capturing
this feedback between ecological and evolutionary dynamics remains a key
challenge.

In this paper we study a continuous-time stochastic model in which
individuals carry one of two traits and evolve under mutation,
selection, and density-dependent regulation. Individuals reproduce with
imperfect inheritance, allowing offspring to differ from their parents,
while mortality depends on trait type, and growth is limited by a global
population ceiling. This leads to a coupled system in which population
size and trait composition evolve together. Despite its simplicity, the
model captures key features of biological systems such as mutation load,
phenotypic switching, and competition between selectively unequal types.

Our analysis identifies threshold conditions governing both persistence
and trait composition. A central finding is that mutation induces an
effective mortality rate that determines whether the population can be
sustained. When reproduction exceeds this rate, the system admits a
positive equilibrium whose composition reflects a balance between
mutation and selection. When inheritance dominates mutation, a second
threshold emerges: population establishment depends not only on overall
growth rates, but also on the initial trait composition, so that some
initial configurations lead to growth while others lead to extinction.
This represents a qualitative departure from classical models in which
persistence is determined solely by net growth rates.

To analyse these dynamics, we combine a deterministic approximation with
a diffusion approximation, both valid in the large-population limit. The
deterministic system captures overall trends, while the stochastic
approximation provides a tractable description of long-lived 
quasi-equilibrium behaviour. The stochastic analysis also reveals a transition
in the sign of trait correlations, reflecting a shift between
fluctuations in total population size and fluctuations in trait
imbalance.

The restriction to two traits is both biologically natural and
mathematically advantageous. Many systems of interest can be reduced to
a binary description, including wild-type versus mutant alleles at a
single locus, as in classical population genetics models
\cite{ewens2004,nowak2006}. Two-type formulations also arise naturally
in microbial and viral evolution, for example in the study of
drug-sensitive versus drug-resistant strains, and in models of
phenotypic switching between active and dormant states
\cite{kussell2005a,kussell2005b}. In adaptive dynamics, even
high-dimensional trait spaces are often effectively reduced to
interactions between a resident population and a single invading mutant
\cite{champagnat2006}. From a mathematical perspective, our
two-dimensional setting permits a detailed analysis of the coupled
dynamics of population size and trait composition, while remaining
sufficiently simple to allow explicit characterisations of equilibria
and threshold behaviour.

The remainder of the paper is organised as follows. In
Section~\ref{section:themodel} we introduce the stochastic model and
derive its deterministic approximation. Subsequent sections analyse key
special cases, characterise equilibria and their stability, and develop
diffusion approximations to describe fluctuations in trait numbers in
quasi equilibrium. These results together provide a unified account of how
mutation, selection, and density dependence shape both persistence and
long-term trait composition.

\section{The model}
\label{section:themodel}

Each individual in the population has one of two traits, $0$ or~$1$. Let
$n_i(t)$ be the number of individuals with trait $i$ at time $t\geq 0$.
Individuals with either trait give birth at the same per-capita rate
$\lambda>0$, the offspring inheriting their trait with probability $q$,
or otherwise changing with probability $1-q$. Without loss of generality
take $0<q\leq 1$. Individuals with trait~$i$ die at per-capita rate
$\mu_i>0$, $i\in \{0,1\}$. We model selection only through the death
rate. Let us take $\mu_1\geq \mu_0$, so that trait~$0$ is selectively
superior. Mutation continually replenishes the inferior type, 
offsetting its selective disadvantage.

Let $(\bn(t),\, t\geq 0)$, with $\bn=(n_0,n_1)$, be a Markov chain in
continuous time taking values in $S=\{\bn\in \Z_+^2: n_0+n_1\leq N\}$,
whose non-zero transition rates are given by
\begin{equation}\label{Q}
\begin{aligned}
q((n_0,n_1),(n_0+1,n_1)) &=\lambda (n_0 q + n_1 (1-q)) (1-(n_0+n_1)/N),
\\
q((n_0,n_1),(n_0,n_1+1)) &=\lambda (n_0 (1-q) + n_1 q) (1-(n_0+n_1)/N),
\\
q((n_0,n_1),(n_0-1,n_1)) &=\mu_0 n_0,
\\
q((n_0,n_1),(n_0,n_1-1)) &=\mu_1 n_1.
\end{aligned}
\end{equation}
The total rate out of state~$\bn$ is
\begin{equation}
q(n_0,n_1)= \lambda (n_0+n_1)(1- (n_0+n_1)/N)+\mu_0 n_0 + \mu_1 n_1.
\label{qn}
\end{equation}
State $\bzero=(0,0)$ is an absorbing state, corresponding to
population extinction. 
It is accessible from all other states, which form a single irreducible 
class. Indeed since $S$ is finite, the extinction state
is reached with probability~$1$.

In the case $\mu_1=\mu_0=\mu$ (no selection), the total number in the
population $s(t)=n_0(t)+n_1(t)$ is a birth-death process taking values
in $\{0,1,\dots,N\}$ with birth rates $\lambda s(1-s/N)$ and death rates
$\mu s$. This is the continuous-time SIS (epidemic) model (Weiss and
Dishon~\cite{WD71}). It provides a useful comparison: although
extinction occurs almost surely, for sufficiently large~$\lambda$, the
system remains near a quasi-equilibrium (endemic) state over long time
scales. This raises the question of whether similar metastable behaviour
arises in the present model, and what determines the resulting trait
composition. One means of addressing these questions is via a
deterministic approximation that is expected to be faithful in the limit
as~$N$ becomes large. An approximating model is easily identified
because the transition rates~(\ref{Q}) are {\em density dependent\/}
(refer to Kurtz~\cite{Kur70}). Following the standard programme outlined
in~\cite{Kur70} (and detailed in Appendix~\ref{appendix2}) we arrive at
the following system of differential equations:
\begin{equation}
\dot{\bx}
= 
F(\bx),\quad \bx\in E,
\label{ODE1}
\end{equation}
where $E=\{\bx\in [0,1]^{\{0,1\}}: x_0+x_1\leq 1\}$, and $F$ has
components
\begin{align}
F_0(\bx)=\lambda(q x_0 + (1-q)x_1)(1-x_0-x_1)-\mu_0 x_0,
\label{F0x}
\\
F_1(\bx)=\lambda((1-q) x_0 + q x_1)(1-x_0-x_1)-\mu_1 x_1,
\label{F1x}
\end{align}
with $x_i(t)$ to be interpreted as the (large-$N$) density of
individuals with trait~$i$ at time~$t$, that is {\em relative to the
population ceiling\/}. In particular, $\bX(t):=\bn(t)/N$ 
can be approximated by
$\bx(t)$ when $N$ is large; a precise statement is given in 
Theorem~\ref{LLN} of Appendix~\ref{appendix2}.

Note that on summing (\ref{F0x}) and (\ref{F1x}) we find that 
$m=x_0+x_1$ satisfies
\begin{equation}
\dot{m}= \lambda m(1-m)-\mu_0 x_0 -\mu_1 x_1,
\label{PKPh}
\end{equation}
which seems very natural; for example, if $\mu_0=\mu_1=\mu$ then (\ref{PKPh})
reduces to Verhulst model~\cite{Ver1838} $\dot{m}=\lambda m(\rho-m)$,
where $\rho=1-\mu/\lambda$, which is known to be a faithful
approximation to the SIS model when~$N$ is large (Theorem~1 of \cite{Pol24}). 

Clearly $\bzero=(0,0)$ is an equilibrium point of (\ref{ODE1}). In fact,
it is the only equilibrium on the boundary of $E$ (on setting
$F_0(\bx)=F_1(\bx)=0$, we see that $x_0=0$ implies $x_1=0$, and vice
versa, and that $x_0+x_1=1$ would entail $x_0=x_1=0$). Furthermore, the
interior of $E$ is attracting, so that solutions to (\ref{ODE1}) which
start in $E$, remain in $E$. Indeed, (\ref{ODE1}) has a unique solution
for all starting states in~$E$. This is explained in
Appendix~\ref{appendix1}.

We shall see that our model exhibits a range of interesting behaviour.
Figure~\ref{fig:range} shows simulations of the stochastic
model~(\ref{Q}), together with trajectories of the deterministic
model~(\ref{ODE1}). These were evaluated using Matlab\footnote{Matlab Version:
25.2.0.3042426 (R2025b) Update 1, Natick, Massachusetts: The MathWorks
Inc.; 2025.} ({\tt ode45}). The numbers relative to $N$ are plotted
against~$t$: trait~$0$ (red), trait~$1$ (blue), and, in plots~(b)
and~(d), the total population (black). A range of behaviour is apparent: (a)
evanescence (numbers die out quickly), (b) coexistence (both traits
survive for long periods), (c) survival of the selectively superior
trait~$0$, and (d) coexistence with the same limiting proportions.
Of particular note is the quasi stationarity apparent in
plots (b)-(d). Notice also in the evanescent case (plot(a)) that the
selectively inferior trait~$1$ dies out and reappears through mutation,
before trait~$0$ eventually disappears.

\begin{figure}[htbp]
    \centering
\includegraphics[width=0.46\textwidth]{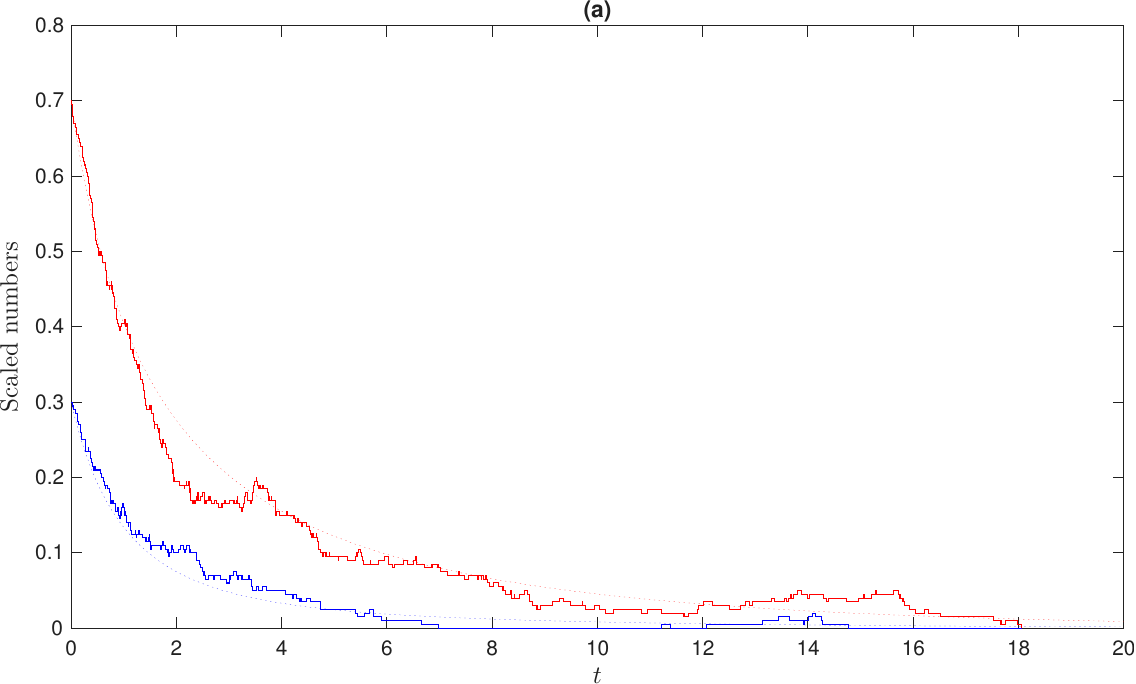}
\includegraphics[width=0.46\textwidth]{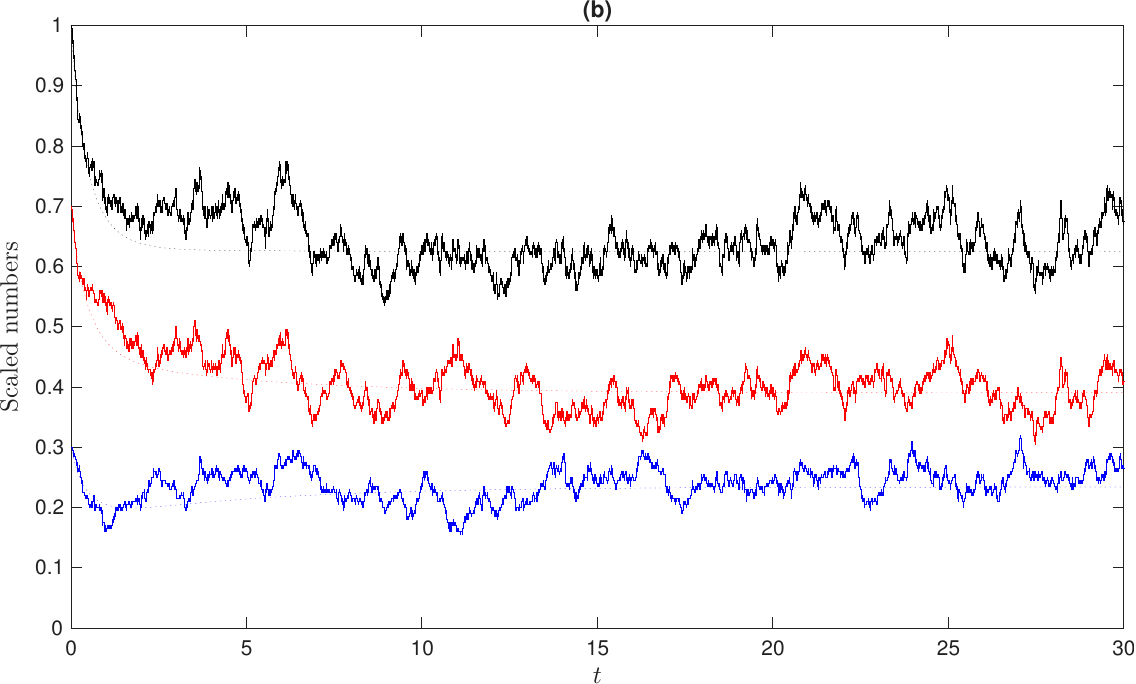}
\includegraphics[width=0.46\textwidth]{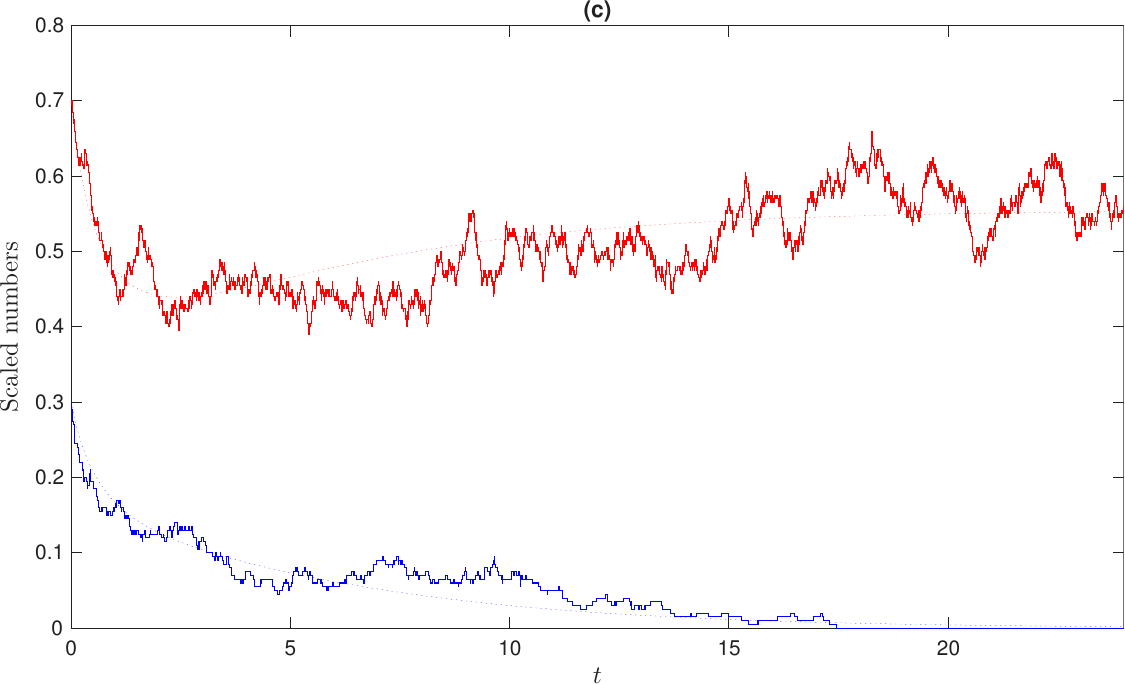}
\includegraphics[width=0.46\textwidth]{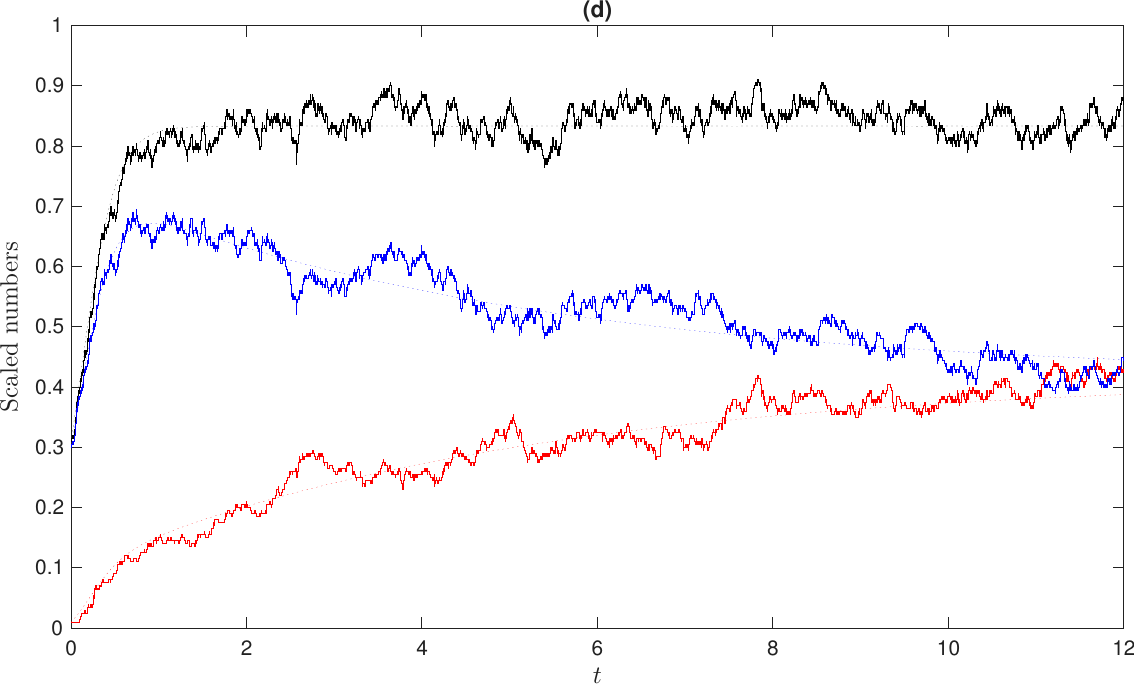}
    \caption{Simulations of numbers relative to $N$,
together with trajectories of the deterministic model:
trait~$0$ (red), trait~$1$ (blue), and, in
plots~(b) and~(d) the total number (black).
All simulations have $N=200$.
Plot (a) has $\lambda = 0.6$, $q = 0.9$, $\mu_0 = 0.7$, $\mu_1 = 1$.
Plot (b) has $\lambda = 2.5$, $q = 0.9$, $\mu_0 = 0.9$, $\mu_1 = 1$.
Plot (c) has $\lambda = 1.8$, $q = 1$, $\mu_0 = 0.8$, $\mu_1 = 1$.
Plot (d) has $\lambda = 6$, $q = 0.9$, $\mu_0=\mu_1 = 1$.}
    \label{fig:range}
\end{figure}

Before addressing questions concerning equilibria and stability,
we will look at some special cases where there are explicit solutions 
to~(\ref{ODE1}).

\subsection{The case $\mu_1=\mu_0=\mu$ (no selection)}
\label{section:noselection}

In this case (\ref{F0x}) and (\ref{F1x}) reduce to
\begin{align*}
F_0(\bx) &=\lambda(q x_0 + (1-q)x_1)(1-m)-\mu x_0,
\\
F_1(\bx) &=\lambda((1-q) x_0 + q x_1)(1-m)-\mu x_1,
\end{align*}
with $m=x_0+x_1$. We assume that $q<1$ here. The case $q=1$
(no mutation) is dealt with on its own in Section~\ref{qequalto1}.
We have already noted that
$m$ follows the Verhulst model
$\dot{m} =\lambda m\left(\rho -m \right)$,
where $\rho=1-\mu/\lambda$. It has the explicit solution
\[
m(t)=
\begin{cases}
\begin{aligned}
&{\ds \frac{\rho m_0}{m_0 + (\rho-m_0)e^{-\lambda\rho t}},}
&\text{if } \lambda\neq \mu,\\[1ex]
&{\ds \frac{m_0}{1+m_0 \mu t},}
&\text{if } \lambda=\mu,
\end{aligned}
\end{cases}
\]
where $m_0=m(0)$ ($\geq 0$).
It follows that $m(t)\to \rho$ when $\lambda>\mu$. Otherwise,
$m(t)\to 0$.
Next we will obtain an ODE for the difference $d(t)=x_1(t)-x_0(t)$, but
before proceeding we note that all points $\bx$ with $x_0=x_1$ form an
invariant manifold; the components follow the Verhulst model with
parameters $\lambda^{\myprime}=2\lambda$ and $\rho^{\myprime}=\frac12
\rho$, and so in particular, if $x_1(0)=x_0(0)=\frac12 m_0$, then
$x_1(t)=x_0(t)=\frac12 m(t)$ ($d(t)=0$ for all $t$). Now, a simple
calculation leads to
$$ 
{\dot d} = (\lambda(2q-1)(1-m)-\mu)d,
$$
and so,
$$
d(t) = d_0 \exp\left( -\mu t + \lambda(2q-1)\int_0^t 1-m(s)\, ds
\right),
$$
where $d_0:=d(0)=x_1(0)-x_0(0)$. 
It is routine then to show that
$$
d(t) 
= 
\begin{cases}
{\ds 
d_0 e^{-\mu t}
\left(
  \frac{\rho e^{\mu t}}{m_0 + (\rho-m_0)e^{-\lambda\rho t}} 
\right)^{2q-1},
}
\qquad &\text{if } \lambda\neq \mu, \\[12pt]
d_0 (1+ m_0 \mu t)^{1-2q} e^{-2(1-q)\mu t},
\qquad &\text{if } \lambda= \mu.
\end {cases}
$$
This encapsulates the ``neutral'' case $q=1/2$, where $\dot{d}=-\mu d$:
$d(t)=d_0e^{-\mu t}$. We may now recover $\bx(t)$ using 
$x_0= \frac12 (m-d)$ and $x_1= \frac12 (m+d)$. 

Notice that if $\lambda=\mu$, then $d(t)\to 0$ whatever the value of $q$. 
If $\lambda>\mu$, then, since $\lambda \rho = \lambda - \mu$ ($\neq 0$),
we have, for large~$t$,
$$
d(t)
\sim
d_0 \left(\frac{\rho}{m_0}\right)^{2q-1}e^{-2\mu (1-q)t},
$$
which implies that $d(t)\to 0$. If $\lambda<\mu$ we get 
$$
d(t)
\sim
d_0 \left(\frac{\rho}{\rho-m_0}\right)^{2q-1}e^{(\lambda(2q-1)-\mu)t},
$$
noting that $\rho<0<m_0$,
which also implies that $d(t)\to 0$, because
$\lambda(2q-1)-\mu < \lambda - \mu <0$. 
Thus, the two traits equalise in the long term, despite continuing
mutation.
We deduce that if $\lambda\leq \mu$,
the $\bx(t)\to\bzero$, while if $\lambda>\mu$, then
$\bx(t)\to ( \frac12 \rho, \frac12 \rho)$, noting that $\rho>0$.
This case is depicted in Figure~\ref{fig:range}d
($\lambda=6$, $\mu_0=\mu_1=1$, $\rho=5/6$).

\subsection{The case $q=1$ (no mutation)}
\label{qequalto1}

In this case (\ref{F0x}) and (\ref{F1x}) reduce to
\begin{align*}
F_0(\bx) &= x_0(\lambda(1-m)-\mu_0) = \lambda x_0 (\rho_0-m),
\\
F_1(\bx) &= x_1(\lambda(1-m)-\mu_1) = \lambda x_1 (\rho_1-m),
\end{align*}
with $m=x_0+x_1$, where $\rho_i=1-\mu_i/\lambda$.
Note that $\rho_1\leq \rho_0$ because
we have assumed that $\mu_1\geq \mu_0$.
We can solve for $\bx(t)$ explicitly.
First notice that if $x_0\equiv 0$, then
$\dot x_1 = \lambda x_1 (\rho_1-x_1)$,
and similarly when $x_1\equiv 0$;
the axes are invariant with the other component following 
the Verhulst model.
So let us assume that $x_0(t)>0$ and let $r(t)=x_1(t)/x_0(t)$. Then,
$$ 
\frac{d}{dt} \log(r) 
= \frac{\dot{x_1}}{x_1} - \frac{\dot{x_0}}{x_0} =
\lambda(\rho_1-\rho_0).
$$
So, $r(t) =r_0 e^{\lambda (\rho_1-\rho_0)t} =r_0 e^{-(\mu_1-\mu_0)t}$,
where $r_0=r(0)=x_1(0)/x_0(0)$. This includes the trivial case
$\mu_0=\mu_1=\mu$, when $r(t)=r_0$ for all $t$. 
Not surprisingly, if $\mu_1>\mu_0$ then $x_1(t)\to 0$.
Now, $x_1=r x_0$, and so $m=(1+rx_0)$,
and therefore, $\dot x_0 = \lambda x_0 (\rho_0-(1+r)x_0)$.
If we set $y(t)=1/x_0(t)$, then $\dot{y} +\lambda \rho_0 y =\lambda(1+r)$,
and so
$$
y(t)= e^{-\lambda\rho_0 t}
\left( y(0) + \lambda \int_0^t e^{\lambda\rho_0 s}(1+r(s)) \, ds\right).
$$
But,
$e^{\lambda\rho_0 s}(1+r(s))
=
e^{\lambda\rho_0 s}+ r_0 e^{\lambda\rho_0 s} e^{\lambda(\rho_1-\rho_0) s}
=
e^{\lambda\rho_0 s}+ r_0 e^{\lambda\rho_1 s}$.
When integrating this, there are special cases to consider: when one
or both of $\rho_0$ and $\rho_1$ is equal to $0$, that is when
$\lambda=\mu_0< \mu_1$ or $\lambda=\mu_1 > \mu_0$, 
or $\lambda=\mu_0=\mu_1$ ($=\mu$). We begin with the generic case.

\begin{itemize}
\item[(1)]
When $\rho_0,\rho_1\neq 0$, that is, $\lambda\neq \mu_0$ 
and $\lambda\neq \mu_1$ (including $\mu_0=\mu_1=\mu$ and $\lambda\neq
\mu$),
$$
x_0(t)= \rho_0 \rho_1 x_0(0)e^{\lambda \rho_0 t} /u(t),
\qquad
x_1(t)= \rho_0 \rho_1 x_1(0) e^{\lambda \rho_1 t}/u(t),
$$
where
$u(t)=\rho_0\rho_1 - \rho_1 x_0(0) (1-e^{\lambda \rho_0 t}) 
-\rho_0 x_1(0) (1-e^{\lambda \rho_1 t})$.

\item[(2)]
When $\rho_1=0$ and $\rho_0>0$, 
that is, when $\lambda=\mu_1>\mu_0$, 
$$
x_0(t) =
\lambda\rho_0 e^{\lambda\rho_0 t} x_0(0)/v_0(t),
\qquad
x_1(t) = \lambda\rho_0 x_1(0)/v_0(t),
$$
where
$v_0(t)=\lambda^2\rho_0 x_1(0) t + \lambda x_0(0) e^{\lambda\rho_0 t}
- \lambda(x_0(0)- \rho_0)$.

\item[(3)]
When $\rho_0=0$ and $\rho_1< 0$, 
that is, when $\lambda=\mu_0< \mu_1$,
$$ 
x_0(t) = \lambda\rho_1 x_0(0)/v_1(t),
x_1(t) = \lambda\rho_1 e^{\lambda\rho_1 t} x_1(0)/v_1(t),
$$
where
$v_1(t)=\lambda^2\rho_1 x_0(0) t + \lambda x_1(0) e^{\lambda\rho_1 t}
- \lambda(x_1(0)- \rho_1)$.

\item[(4)]
When $\rho_0=\rho_1=0$, that is, $\lambda=\mu_0=\mu_1$, 
$$
x_0(t) = \frac{x_0(0)}{1+\lambda t (x_0(0) + x_1(0))}, 
\qquad
x_1(t) = \frac{x_1(0)}{1+\lambda t (x_0(0) + x_1(0))}.
$$
\end{itemize}

\noindent
There is a variety of limiting behaviour.
\begin{itemize}
\item[(a)]
If $\lambda\leq \mu_0$ ($\leq \mu_1$), then $\bx(t)\to(0,0)$.
\item[(b)]
If $\lambda>\mu_0$ and $\mu_1>\mu_0$, then $\bx(t)\to (\rho_0,0)$
($\rho_0=1-\mu_0/\lambda >0$).
This case is depicted in Figure~\ref{fig:range}c
($\mu_0 = 0.8$, $\lambda = 1.8$ $\mu_1 = 1$, $\rho_0=5/9$).
\item[(d)]
If $\mu_0=\mu_1$ ($=\mu)$, and $\lambda>\mu$,
$$
\bx(t)\to 
\left(
\frac{\rho x_0(0)}{x_0(0)+ x_1(0)},
\frac{\rho x_1(0)}{x_0(0)+ x_1(0)}
\right).
$$
\end{itemize}
This makes sense. Trait~$1$ is certain to die out,
reflecting its selective disadvantage ($\mu_1 \ge \mu_0$).
If the death rate trait~$0$
is sufficiently small it will survive; otherwise it will also die out.
In Case (d) (no mutation or selection) the traits coexist, and
the limit point lies on the
line $\{x\in [0,1]^2: x_0+x_1=\rho\}$. Moreover, 
if $(x_0^*,x_1^*)$ is any particular point on that line, its domain
of attraction is the line
$\{(\beta x_0^*/\rho,\beta x_1^*/\rho),\, 0< \beta\leq 1\}$.
Indeed all straight lines passing through 
$\bzero$ with positive slope are invariant manifolds.

We now turn to the analysis of equilibria and their stability.

\section{Equilibria}

Whilst it seems impossible to obtain an explicit solution
to~(\ref{ODE1}) in all but the cases considered above, we can identify
all equilibria and classify them. We show that the system admits at most
one biologically relevant interior equilibrium, whose existence and
stability is determined by threshold conditions. We proceed in three
steps: (i) derive candidate equilibria, (ii) identify those lying in~$E$, 
and (iii) classify their stability.
We show that persistence is governed by an effective mortality $\mu^*$, 
lying between $\mu_0$ and $\mu_1$, which depends explicitly on the mutation 
and selection parameters.

Clearly $\bzero$ satisfies the equilibrium equations
\begin{align}
\lambda(q x_0 + (1-q)x_1)(1-m)&=\mu_0 x_0,
\label{EQ1}
\\
\lambda((1-q) x_0 + q x_1)(1-m)&=\mu_1 x_1,
\label{EQ2}
\end{align}
where $m=x_0+x_1$;
as already noted, it is the only boundary equilibrium.
So, assume that $x_0\neq 0$ and $x_1\neq 0$, and note
that $m\neq 1$ (since if $m=1$ then 
(\ref{EQ1}) and (\ref{EQ2}) would imply
$\mu_0 x_0=\mu_1 x_1=0$, and hence $x_0=x_1=0$, a contradiction).
We will see that, in all but two exceptional
cases, there is at least one non-zero solution---usually two solutions---to
(\ref{EQ1}) and~(\ref{EQ2}), at most one of which lies in $E$,
and that all can be exhibited explicitly.
We then classify the equilibria by way of the Jacobian
\begin{equation}
B(\bx) := \nabla F(\bx)=
\left(
\begin{matrix}
\lambda(q(1-2x_0)-x_1)-\mu_0 & \lambda ((1-q)(1-2x_1)-x_0) \\
\lambda((1-q)(1-2x_0)-x_1) & \lambda(q(1-2x_1)-x_0)-\mu_1  
\end{matrix}
\right).
\label{Jac}
\end{equation}
If $\bx^*=(x_0^*,x_1^*)$ is any specific equilibrium point, its
classification is determined by the signs of the eigenvalues of $B(\bx^*)$.
Furthermore, we may write $F(\bx)=B(\bx^*)(\bx-\bx^*)+h(\bx-\bx^*)$, where 
$$
h(\bx)
=
-\lambda
\left(
\begin{matrix}
(q x_0 + (1-q)x_1)(x_0+x_1) \\
((1-q) x_0 + qx_1)(x_0+x_1)
\end{matrix}
\right).
$$
Since $h$ satisfies $h(\bx)=\mathcal{O}(\|\bx\|^2)$ as $\|\bx\|\to 0$,
we may use Theorems~10.14 and~10.15 of Jordan and Smith \cite{JS07},
which apply to two-dimensional systems, and
connect the classification of $\bx^*$ with
its classification for the linearized system $\dot \bx = B(\bx^*)(\bx-\bx^*)$.
They say, respectively,
(i)~that $\bx^*$ is asymptotically stable when it is asymptotically stable
for the linearized system, and
(ii)~if the eigenvalues of $B(\bx^*)$ are different, non-zero, and
at least one has positive real part, then $\bx^*$ is unstable.
Furthermore, the Hartman-Grobman Theorem (Theorem~1.40 of~\cite{Chi2024})
applies when the eigenvalues of $B(\bx^*)$ are non-zero. It
says that the linearized system has the same phase portrait as the
original system in a sufficiently small neighbourhood of $\bx^*$.
It turns out that we get the {\em same\/} $h$ for all equilibria,
as a consequence of $h$ being quadratic (this is explained in
Appendix~\ref{appendix5}).
When one of the eigenvalues of $B(\bx^*)$ is equal to $0$,
which happens when $\det(B(\bx^*))=0$, the classification is
inconclusive.

In preparation for the general case,
it will be instructive to consider the earlier cases where we
have explicit formulae for the trajectories.

\subsection{The case $\mu_1=\mu_0=\mu$ (no selection)}

Assume that $q<1$. The equilibrium equations become
\begin{align}
\lambda(q x_0 + (1-q)x_1)(1-m)&=\mu x_0,
\label{F0xB}
\\
\lambda((1-q) x_0 + q x_1)(1-m)&=\mu x_1,
\label{F1xB}
\end{align}
and, on dividing (\ref{F0xB}) by (\ref{F1xB}), we notice that
for any non-zero solution $(x_0,x_1)$, the ratio
$r=x_1/x_0$ satisfies $r^2=1$. 
When $r=1$ we get 
$\bx=(\rho/2,\rho/2)$, where recall that $\rho=1-\mu/\lambda$. 
When $r=-1$, it is necessary that
$\lambda=\mu/(2q-1)$ and $q>1/2$, in which case all members of
the invariant manifold $\M:=\{\bx\in \R^2: x_0=-x_1\}$
are equilibria, but we exclude this from practical
considerations. As we shall see below, this turns out to be a
transitional case in the classification of $\bzero$.

The Jacobian of $F$ at $\bzero$,
$$
B(\bzero)=
\left(
\begin{matrix}
\lambda q-\mu & \lambda (1-q) \\
\lambda (1-q) &\lambda q-\mu  
\end{matrix}
\right),
$$
has eigenvalues $\alpha_1=\lambda (2q-1) - \mu$
and $\alpha_2=\lambda-\mu$, noting that $\alpha_1<\alpha_2$ since $q<1$.
So, we have the following classification.

\begin{itemize}
\item[(i)] If $\lambda<\mu$, then $\bzero$ is a stable node.
\item[(ii)] If $q<1/2$ and $\lambda>\mu$, or if
$q>1/2$ and $\mu<\lambda<\mu/(2q-1)$, then $\bzero$ is a saddle point.
\item[(iii)] If $q>1/2$ and $\lambda>\mu/(2q-1)$, then $\bzero$ is an unstable node.
\end{itemize}
The classification is inconclusive
at the transitional values, $\lambda=\mu/(2q-1)$ and $\lambda=\mu$
(where $\alpha_1=0$, respectively $\alpha_2=0$).
Numerical evidence suggests that
$\bzero$ is a saddle point when $\lambda=\mu$
and an unstable node when $\lambda=\mu/(2q-1)$ ($q>1/2$).
Refer to Figure~\ref{fig:muconstant}, which illustrates these classifications.

In cases (ii) and (iii), where $\lambda>\mu$, the non-zero equilibrium 
$\bx^*=(\rho/2,\rho/2)$ lies in~$E$, and
$$
B(\bx^*)=
\left(
\begin{matrix}
\mu q - (\lambda+\mu)/2, &\mu (1-q) - (\lambda-\mu)/2 \\
\mu (1-q) - (\lambda-\mu)/2  &\mu q - (\lambda+\mu)/2
\end{matrix}
\right).
$$
One may verify that the eigenvalues of $B(\bx^*)$ are
$\alpha_1=-2\mu(1-q)< \alpha_2=\mu-\lambda<0$, so
$\bx^*$ is a stable node.
Note that if $\lambda<\mu$, then $(\rho/2,\rho/2)$ is a 
saddle point in the third quadrant, but not biologically relevant.

Figure~\ref{fig:muconstant}
illustrates the various cases, all with $\mu=3$.
Trajectories are shown for various starting points (green).
The red dots indicate equilibria.
Plots (a)--(e) have $q>\frac12$ and increasing values of $\lambda$,
corresponding to
(a)~$\lambda<\mu$,
(b)~$\lambda=\mu$,
(c)~$\mu < \lambda < \mu/(2q-1)$,
(d)~$\lambda = \mu/(2q-1)$, and
(e)~$\lambda > \mu/(2q-1)$.
Plot~(f) has $q<1/2$ and $\lambda>\mu$. 
In cases (c)--(f), the stable equilibrium 
$\bx^*=(\rho/2,\rho/2)$ has $\rho/2$ equal to $0.2$, $1/3$,
$0.35$, and $0.35$, respectively. 
The two starting points that lie outside $E$ are included
to help elucidate the behaviour of the system near~$\bzero$.
For example, compare (e) and (f), both of which have $\lambda=10$.
In case (e) $\bzero$ is an unstable node, 
while in case (f) $\bzero$ is a saddle point.

\begin{figure}[htbp]
    \centering
\includegraphics[width=0.32\textwidth]{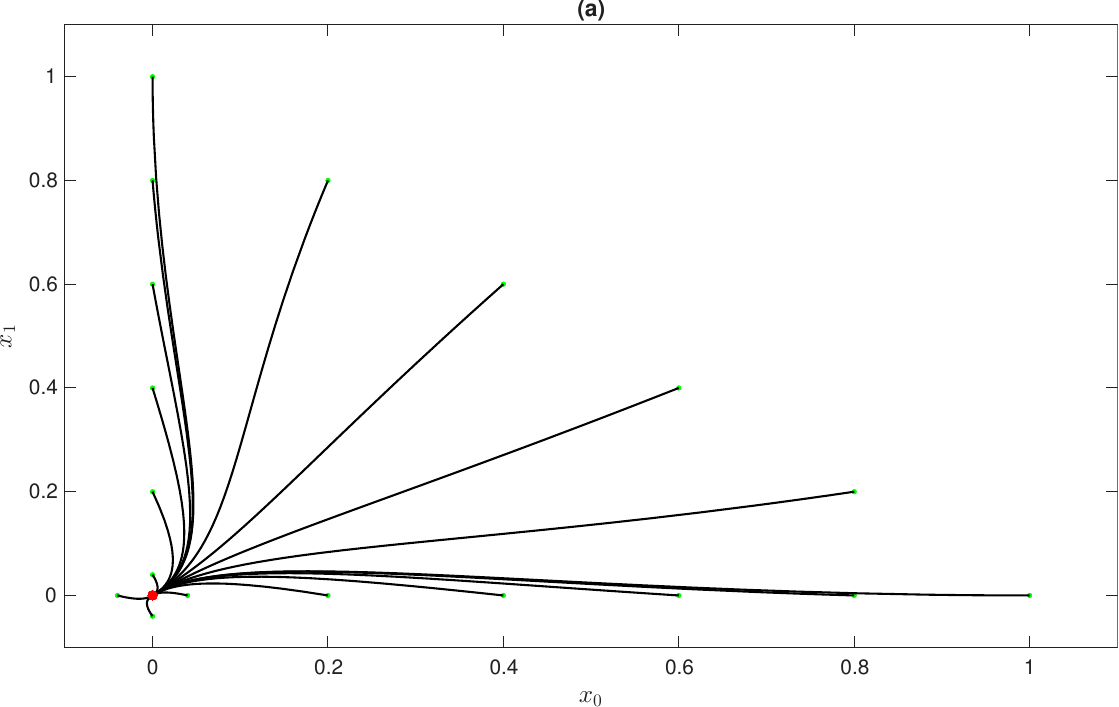}
\includegraphics[width=0.32\textwidth]{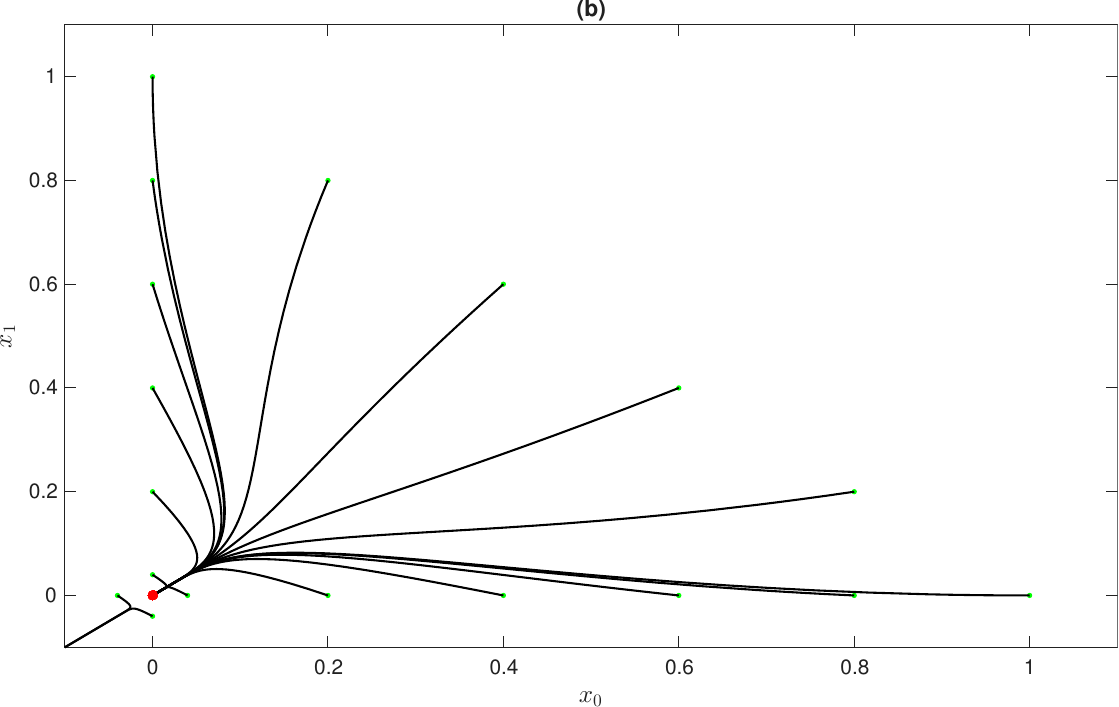}
\includegraphics[width=0.32\textwidth]{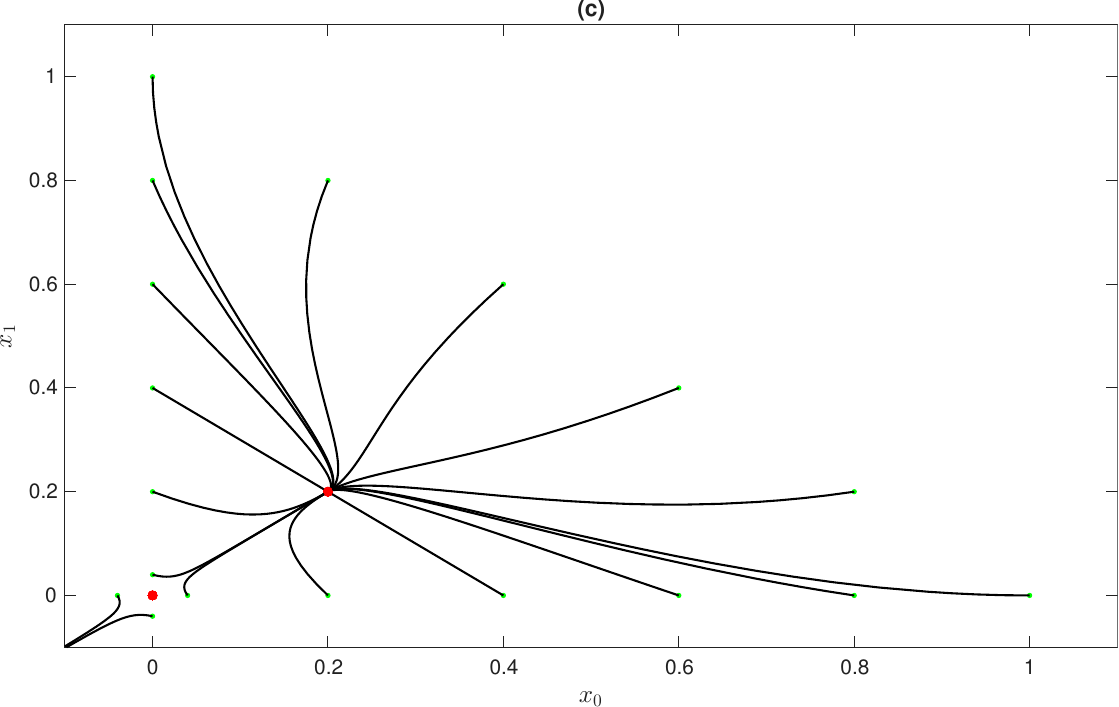}
\includegraphics[width=0.32\textwidth]{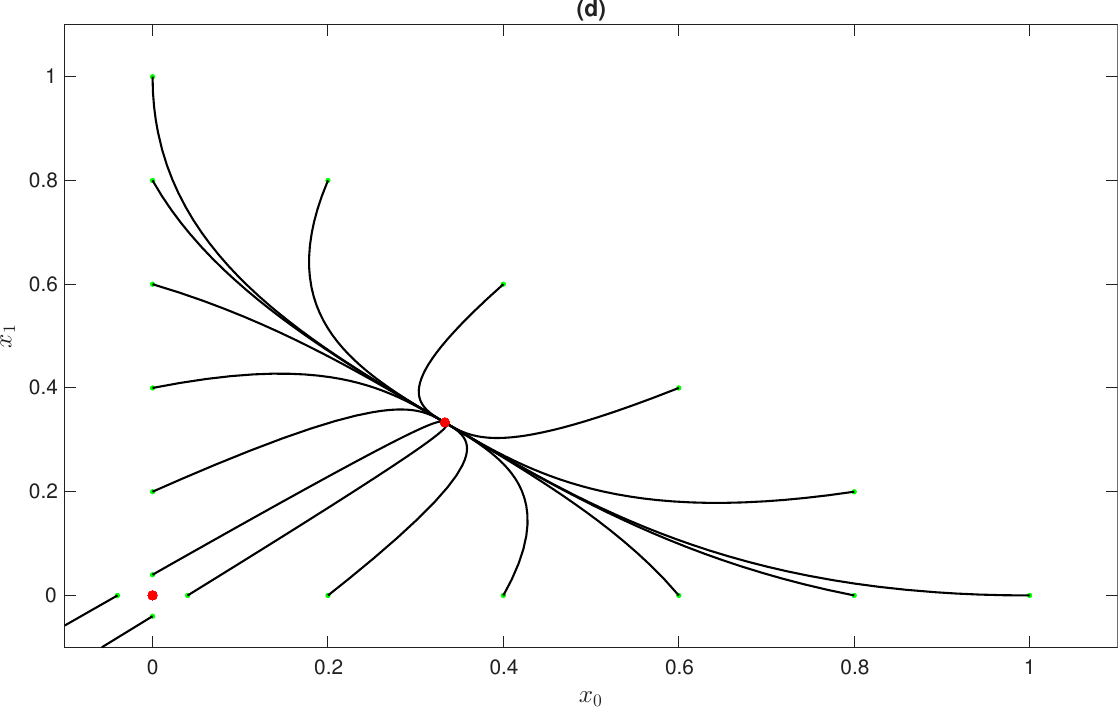}
\includegraphics[width=0.32\textwidth]{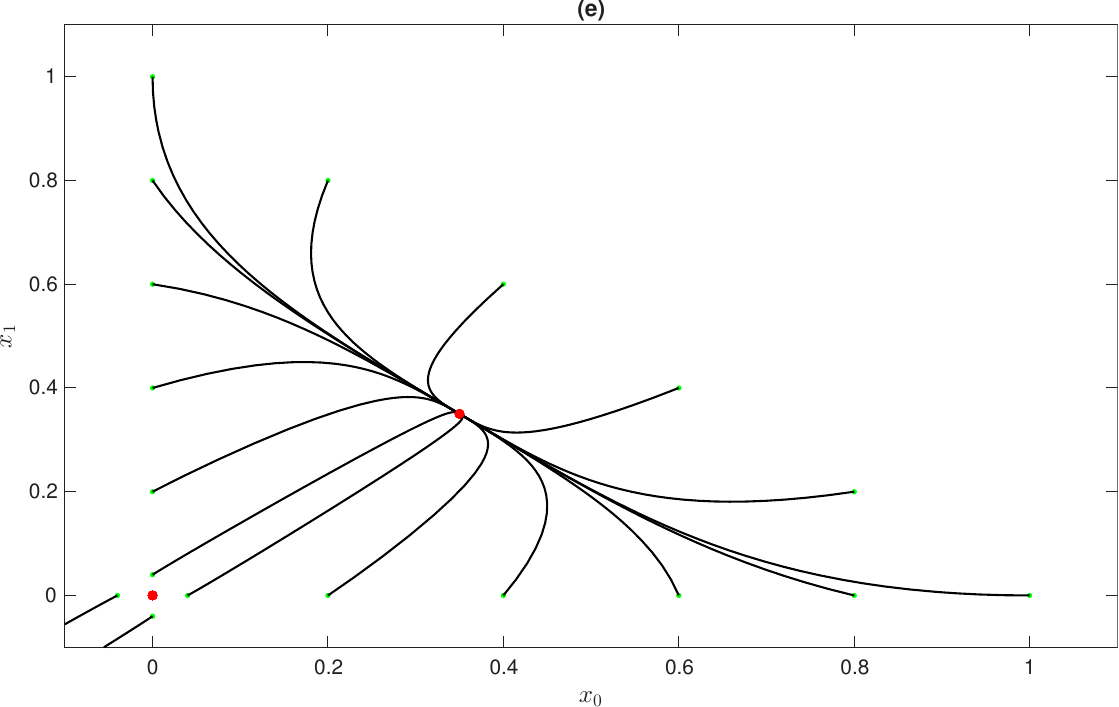}
\includegraphics[width=0.32\textwidth]{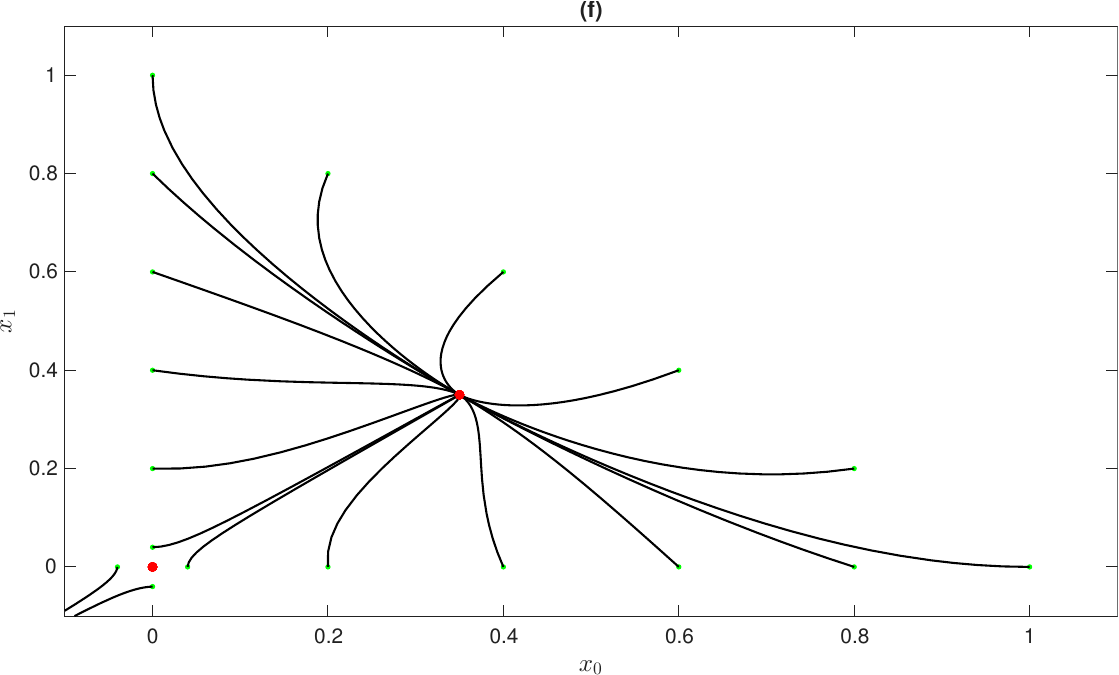}
    \caption{Trajectories of the deterministic model with $\mu_0=\mu_1=3$. 
Plots (a)--(e) have $q=\frac23$ and increasing values of $\lambda$:
(a) $\lambda=2$, (b) $\lambda=3$, (c) $\lambda=5$, (d) $\lambda=9$, and 
(e) $\lambda=10$. 
Plot~(f) has $q=\frac13$ and $\lambda=10$. 
The green dots indicate starting points.
The red dots indicate equilibria.}
    \label{fig:muconstant}
\end{figure}

\subsection{The case $q=1$ (no mutation)}

When $q=1$, (\ref{EQ1}) and (\ref{EQ2}) become
$\lambda (\rho_0-m) x_0 = 0$ and
$\lambda (\rho_1-m) x_1 = 0$, where recall that
$\rho_i=1-\mu_0/\lambda$, $i=0,1$. So, there are four cases to examine.

\begin{itemize}
\item[(1)]
$\bx^*=\bzero$:
$$
B(0,0)=
\left(
\begin{matrix}
\lambda\rho_0 &0 \\
0             & \lambda\rho_1
\end{matrix}
\right)
=
\left(
\begin{matrix}
\lambda-\mu_0 &0 \\
0             & \lambda-\mu_1
\end{matrix}
\right)
$$
has eigenvalues $\alpha_1= \lambda-\mu_1$ and $\alpha_2= \lambda-\mu_0$,
noting that $\alpha_1\leq \alpha_2$ since $\mu_1\geq \mu_0$.
So, $\bzero$ is a stable node if $\lambda<\mu_0$,
it is a saddle point if $\mu_0<\lambda<\mu_1$,
and it is an unstable node if $\lambda>\mu_1$.
Numerical evidence suggests that
$\bzero$ is a saddle point when $\lambda=\mu_0$
and an unstable node when $\lambda=\mu_1$. Refer to Figure~\ref{fig:qequals1},
which illustrates these classifications.

\item[(2)]
$\bx^*=(\rho_0,0)$, with $\rho_0>0$ ($\lambda>\mu_0$):
$$
B(\rho_0,0)=
\left(
\begin{matrix}
-\lambda\rho_0 & -\lambda \rho_0 \\
0              & \lambda(\rho_1-\rho_0)
\end{matrix}
\right)
=
\left(
\begin{matrix}
\mu_0-\lambda & \mu_0-\lambda \\
0             & \mu_0-\mu_1
\end{matrix}
\right)
$$
has eigenvalues $\alpha_1=\mu_0-\lambda$ ($<0$) and
$\alpha_2=\mu_0-\mu_1$ ($\leq 0$).
So, $(\rho_0,0)$ is a stable node if $\mu_0<\mu_1$.
(It cannot be a saddle point or an unstable node for the
linearized system.)
Refer to plots~(c), (d), and (e) in 
Figure~\ref{fig:qequals1}.

\item[(3)]
$\bx^*=(0,\rho_1)$ with $\rho_1>0$ ($\lambda>\mu_1$):
Similarly, the eigenvalues of $B(0,\rho_1)$
are $\alpha_1=\mu_1-\lambda$ ($<0$) and 
$\alpha_2=\mu_1-\mu_0$ ($\geq 0$). So, $(0,\rho_1)$ is an unstable node
if $\mu_1>\mu_0$.
(It cannot be a saddle point or a stable node for the
linearized system.) Refer to plot~(e) in 
Figure~\ref{fig:qequals1}.

\item[(4)]
$x_0^*,x_1^*>0$: 
We must
have $x_0^*+x_1^*=\rho_0=\rho_1$, so it is necessary that $\mu_0=\mu_1$
($=\mu$), $\lambda>\mu$, and $x_0^*+x_1^*=\rho$, where $\rho=1-\mu/\lambda$.
Then,
$$
B(x_0^*,x_1^*)
=
\left(
\begin{matrix}
-\lambda x_0^* & -\lambda x_0^* \\
-\lambda x_1^* & -\lambda x_1^* 
\end{matrix}
\right),
$$
has eigenvalues $\alpha_1=0$ and 
$\alpha_2=-\lambda(x_0^*+x_1^*)=-\lambda\rho=\mu-\lambda<0$.
Whilst we have a zero eigenvalue, we have already seen
that the line $\{\bx\in [0,1]^2: x_0+x_1=\rho\}$ is attracting.
\end{itemize}

\begin{figure}[htbp]
   \centering
   \includegraphics[width=0.32\textwidth]{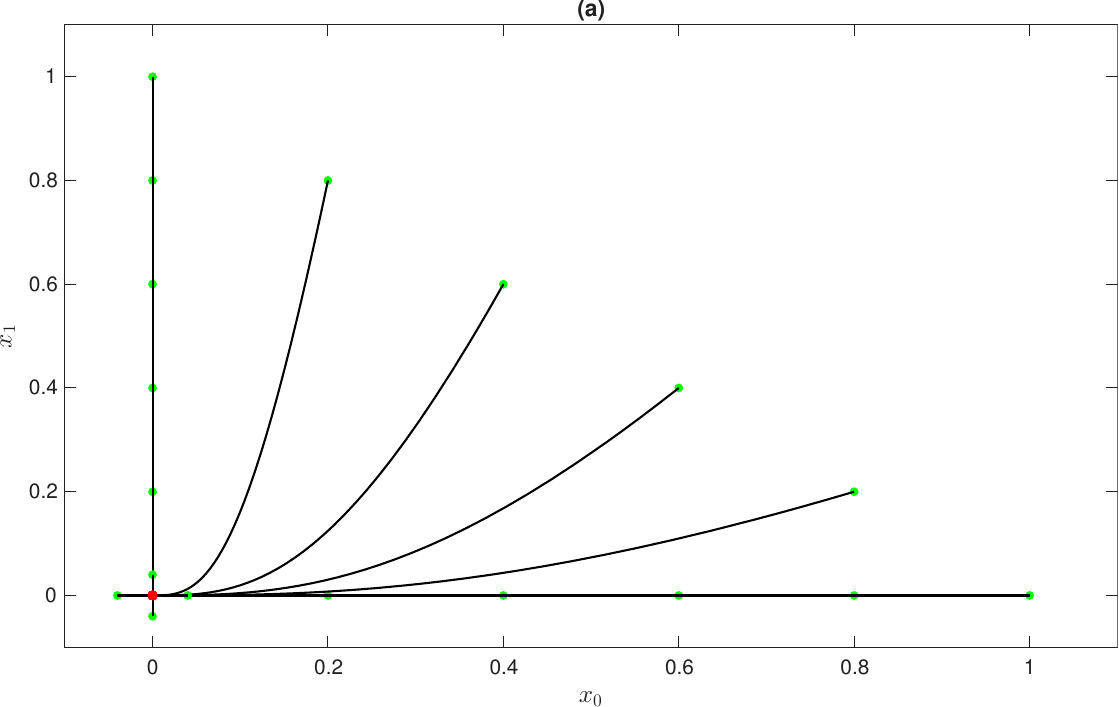}
   \includegraphics[width=0.32\textwidth]{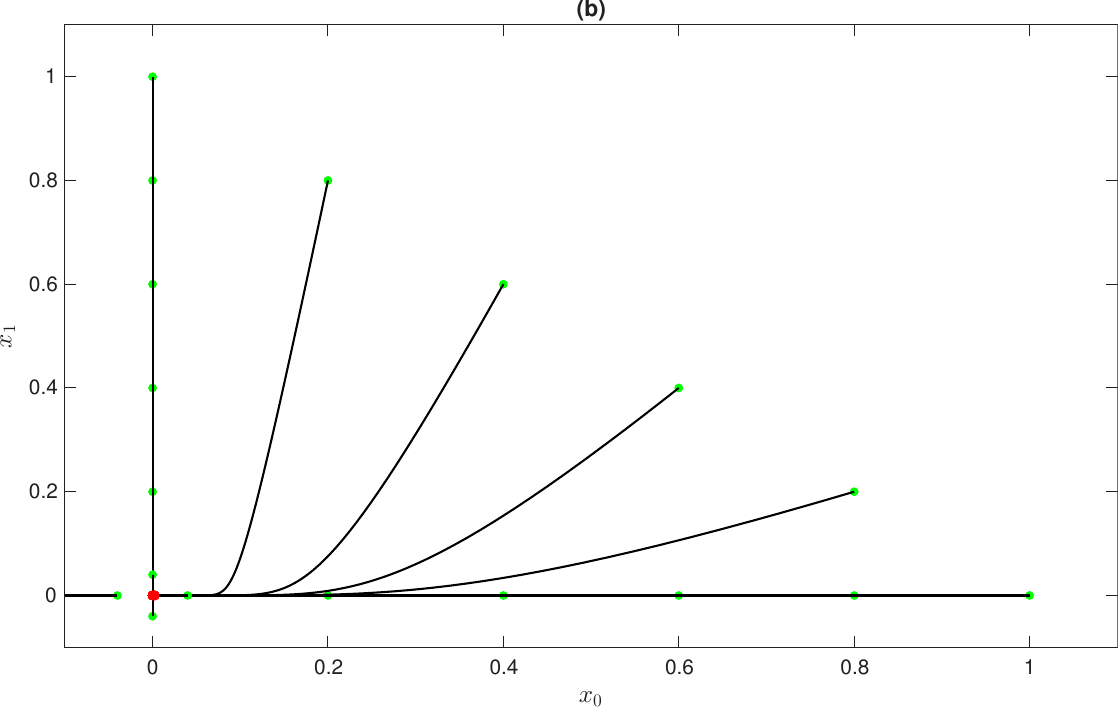}
   \includegraphics[width=0.32\textwidth]{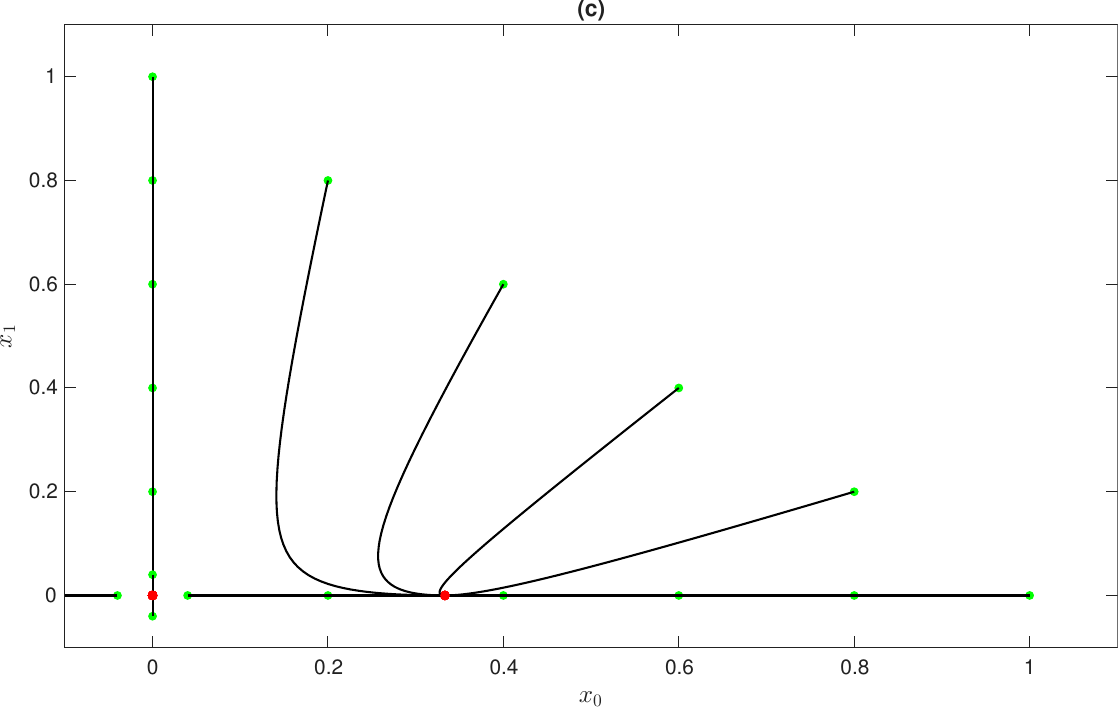}
   \includegraphics[width=0.32\textwidth]{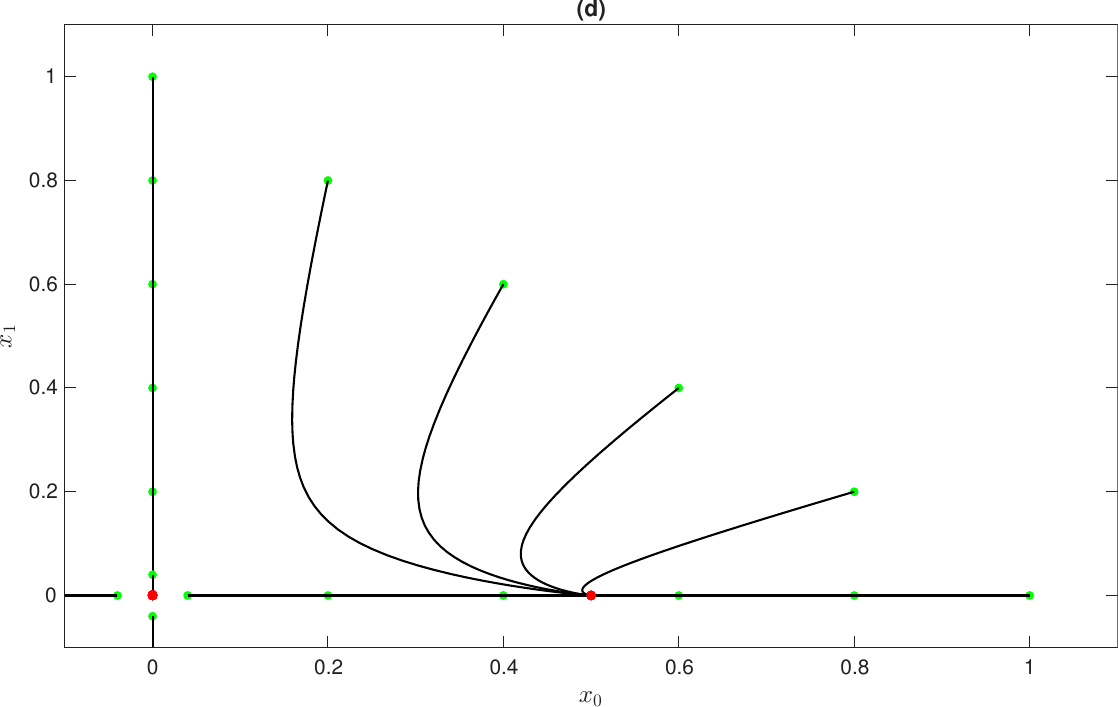}
   \includegraphics[width=0.32\textwidth]{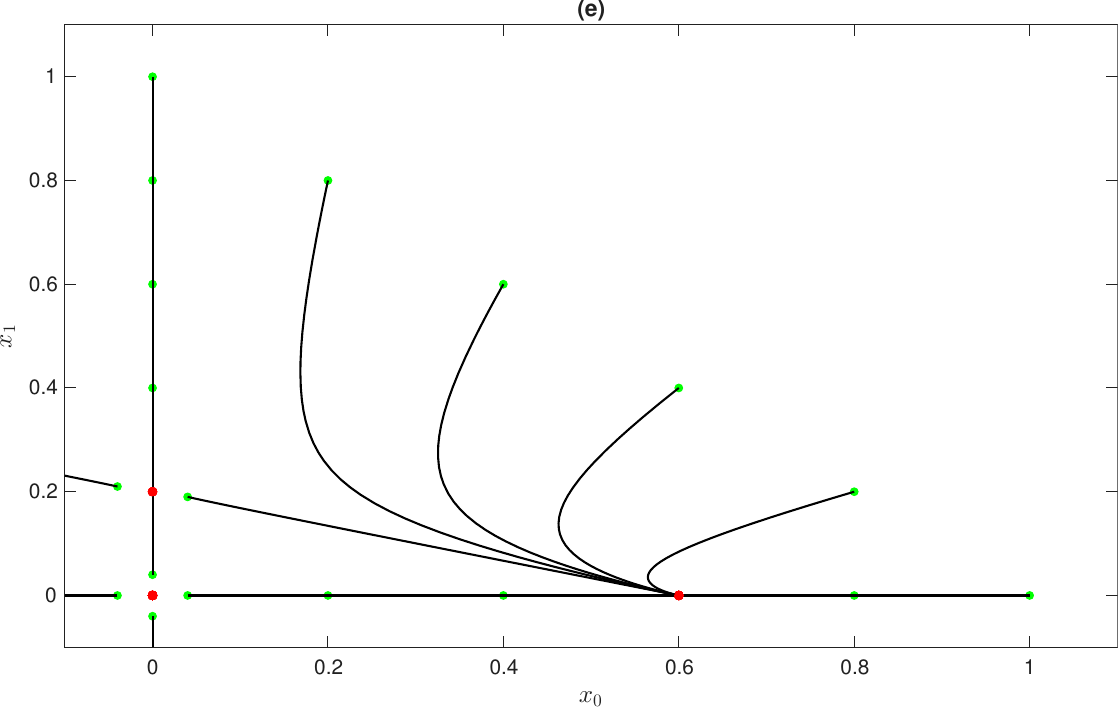}
   \includegraphics[width=0.32\textwidth]{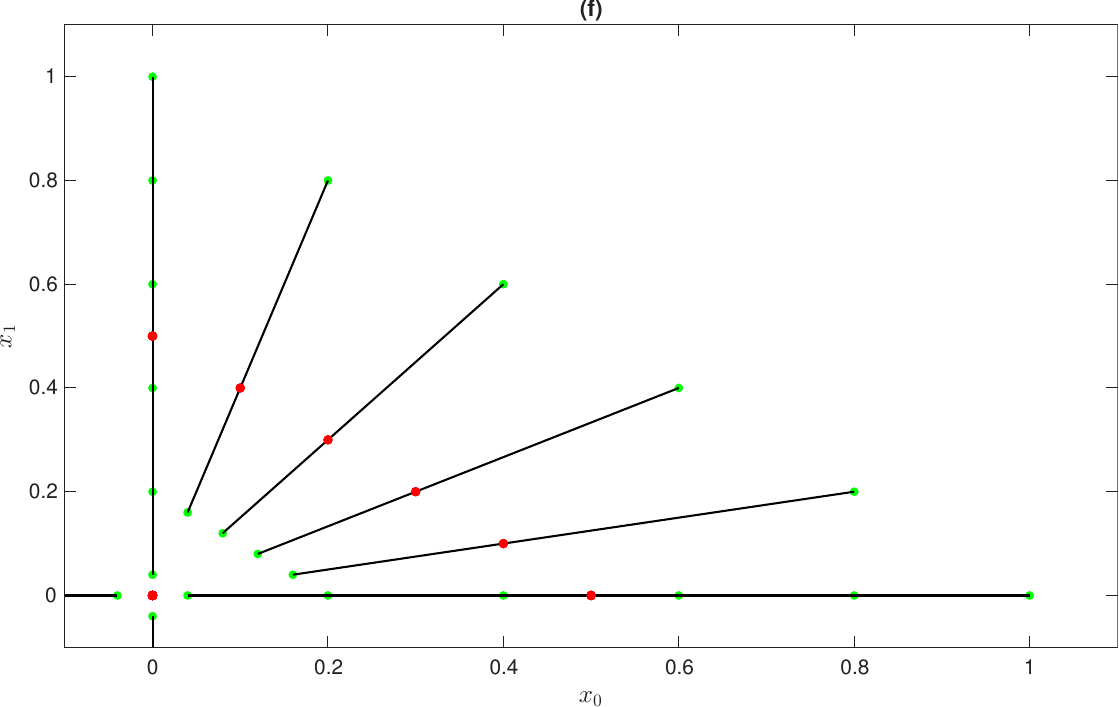}
   \caption{Trajectories of the deterministic model with $q=1$. Plots 
     (a)--(e) have $\mu_0=2$ and $\mu_1=4$, and increasing values 
     of $\lambda$: (a) $\lambda=1$, (b) $\lambda=2$, (c) $\lambda=3$, 
     (d) $\lambda=5$, and (e) $\lambda=5$. Plot (f) has 
     $\mu_0=\mu_1=2$ and $\lambda=4$.  The green dots indicate 
     starting points. The red dots indicate equilibria.}
  \label{fig:qequals1}
\end{figure}

Figure~\ref{fig:qequals1}
illustrates the various cases.
Trajectories are shown for various starting points (green).
The red dots indicate equilibria.
Plots (a)--(e) have $\mu_0=2$ and $\mu_1=4$, 
and increasing values of $\lambda$, corresponding to
(a) $\lambda<\mu_0$,
(b) $\lambda=\mu_0$,
(c) $\mu_0 < \lambda < \mu_1$,
(d) $\lambda = \mu_1$, and
(e) $\lambda > \mu_1$.
Plot (f) has $\mu_0=\mu_1=\mu$ and $\lambda>\mu$.
A few starting points that lie outside $E$ are included
to help elucidate the behaviour of the system near equilibria
(unstable node/saddle point/unstable node).
Notice the appearance in plot~(e) of the equilibrium $(0,\rho_1)$
(unstable), $\rho_1=0.2$. Notice also in plot~(f) the line of equilibria
$\{\bx\in [0,1]^2: x_0+x_1=\rho\}$ ($\rho=1/2$).

\subsection{The generic case $\mu_1>\mu_0$ and $q<1$}

%
To obtain non-zero equilibria, divide (\ref{EQ1}) by~(\ref{EQ2}) 
and rearrange to get
$$
(1-q)x_1^2 +q(1-s)x_0x_1 -(1-q)s x_0^2=0,
$$
where $s=\mu_0/\mu_1$. So, $r:=x_1/x_0$ satisfies 
\begin{equation}
f(r)=0, \quad \text{where }
f(r)=(1-q) r^2 +q(1-s)r -(1-q)s.
\label{eq:r}
\end{equation}
This has roots
\begin{equation}
r^{\pm} = \frac{-q(1-s) \pm \sqrt{\delta}}{2(1-q)},
\label{rpm}
\end{equation}
where
\begin{equation}
\delta= q^2 (1-s)^2 + 4(1-q)^2 s
= q^2 (1+s)^2 - 4(2q-1)s \ (>0).
\label{delta}
\end{equation}
Since $1-q$ and $(1-q)s$ are both strictly positive, the graph of $f$ is an 
upward opening parabola with negative intercept, and so  $r^+>0$ and
$r^-<0$. 

It is clear that the choice $q=1/2$ (when an offspring's trait is 
independent of its parent's) leads to considerable simplification,
because from (\ref{EQ1}) and (\ref{EQ2}) we get
$\mu_0 x_0=\mu_1 x_1$ (corresponding to the
positive root $r^+=s=\mu_0/\mu_1$), and hence $m=(1+s)x_0$. 
On substituting this into
$\lambda m(1-m)= \mu_0 x_0+\mu_1 x_1$ (which follows on summing (\ref{EQ1}) 
and (\ref{EQ2})), we obtain an explicit expression 
for the non-zero equilibrium $\bx^*=(x_0^*,x_1^*)$. It has a pleasingly simple
form. Let $\theta^*=1/(1+s)$. Then,
\begin{equation}
x_0^*
=\theta^* m^*
=\frac{\mu_1}{\mu_0+\mu_1} m^*,
\quad
x_1^*
=(1-\theta^*) m^*
=\frac{\mu_0}{\mu_0+\mu_1} m^*,
\quad \text{and} \quad
m^*=1-\frac{\mu^*}{\lambda},
\label{EQmuconstant}
\end{equation}
where $\mu^*=2\mu_0\mu_1/(\mu_0+\mu_1)$, being the
{\em harmonic mean\/} of $\mu_0$ and $\mu_1$.
This anticipates the general result, where $\mu^*$ is identified,
and given explicitly in terms of the positive root $r^+$.
If $\lambda=\mu^*$, then $\bzero$ is the {\em only\/} equilibrium point.
If $\lambda<\mu^*$, then the equilibrium would lie in the third
quadrant, so we exclude it from practical consideration.
If $\lambda>\mu^*$, then $\bx^*=(x_0^*,x_1^*)\in E$; in this case
$m^*$ can be interpreted as the carrying capacity of the population,
which is apportioned to each trait in accordance with the relative
per-capita death rates.

Let us classify the equilibria. We see that
$$
B(x_0,x_1)=
\left(
\begin{matrix}
a-\mu_0 & a \\
a & a-\mu_1  
\end{matrix}
\right),
\qquad \text{where } a=\lambda(\tfrac12-(x_0+x_1)),
$$
noting that $B$ depends on $x_0$ and $x_1$ only through their sum.
Also $\tr$, the trace of $B$, is equal to $2a-(\mu_0+\mu_1)$, 
and $d:=\det(B)=\mu_0\mu_1 - (\mu_0+\mu_1)a$.
So, at $\bx=\bzero$, we have $a=\tfrac12 \lambda$. Thus, 
$$
\tr=\lambda-(\mu_0+\mu_1), 
\quad\text{and}\quad
d= \mu_0\mu_1 - \tfrac12 \lambda (\mu_0+\mu_1)
=\tfrac12 (\mu_0+\mu_1)(\mu^*-\lambda).
$$
The eigenvalues of $B(\bzero)$ are real because
the discriminant $\Delta:=\tr^2-4d=\lambda^2 + (\mu_1-\mu_0)^2$ is strictly
positive.
It follows that $\bzero$ is a stable node if $\tr<0$ and $d>0$,
that is, if $\lambda<\mu^*$, 
since $\mu^*\leq \mu_0+\mu_1$ (because 
harmonic mean $\leq$ arithmetic mean).
It is a saddle point if $d<0$,
that is $\lambda>\mu^*$. It cannot be an unstable node, for this 
would require $\tr>0$ and $d>0$. 
The classification of $\bzero$ is inconclusive when $\lambda=\mu^*$,
but numerical evidence suggests that
$\bzero$ is a saddle point in this case.

At $\bx=\bx^*$, with $\lambda>\mu^*$,
we get $a=\lambda(\tfrac12-m^*)=\mu^*-\frac12 \lambda$. So, 
$$
\tr=2\mu^*-(\mu_0+\mu_1) - \lambda
= -\frac{(\mu_0-\mu_1)^2}{\mu_0+\mu_1} - \lambda \ (<0 \text{ always}),
$$
and
$$
d=\mu_0\mu_1 - (\mu_0+\mu_1)(\mu^*-\tfrac12 \lambda)
=\tfrac12 (\mu_0+\mu_1)(\lambda-\mu^*)>0.
$$
It follows that $\bx^*$ is a stable node. 

Figure~\ref{fig:qequalshalf} illustrates the various cases.
Trajectories are shown for various starting points (green).
The red dots indicate equilibria.
We have $\mu_0=2$ and $\mu_1=4$, 
and increasing values of $\lambda$, corresponding to
(a) $\lambda<\mu^*$,
(b) $\lambda=\mu^*$,
(c) $\lambda>\mu^*$.
Here $\mu^*=\frac83$, $m^*=\frac12$, and so $x_0^*=\frac23 m^*=\frac13$, and 
$x_1^*=\frac13 m^*=\frac16$.
Notice that $\bzero$ is a saddle point when $\lambda=\mu^*$.

\begin{figure}[htbp]
    \centering
\includegraphics[width=0.32\textwidth]{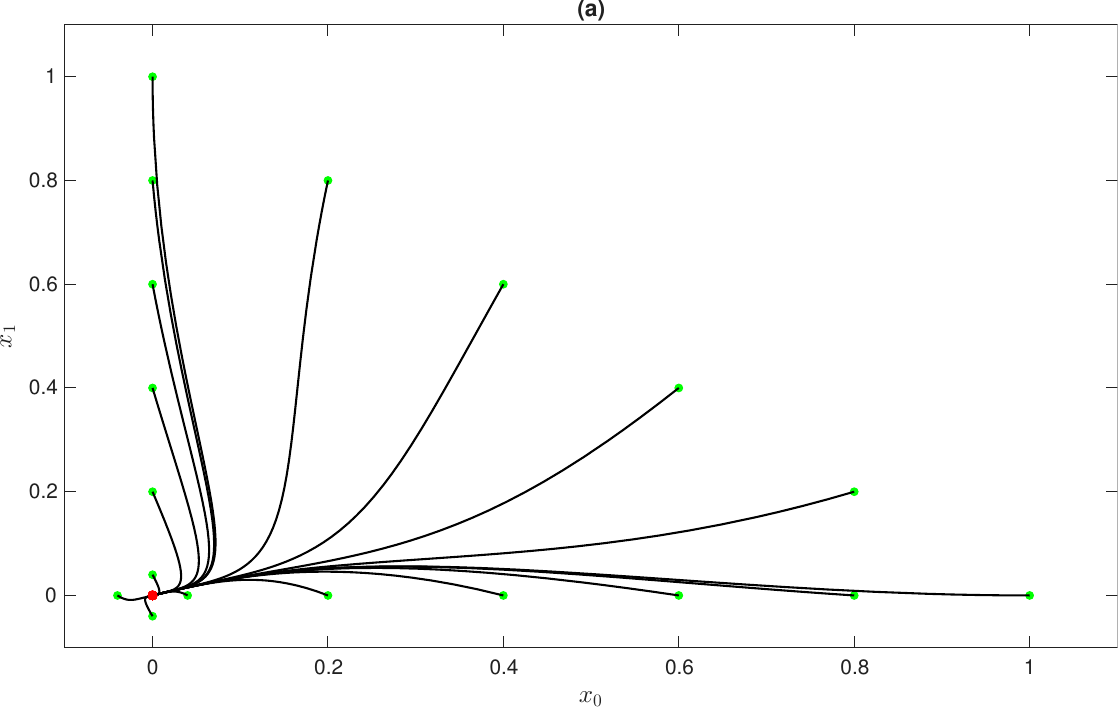}
\includegraphics[width=0.32\textwidth]{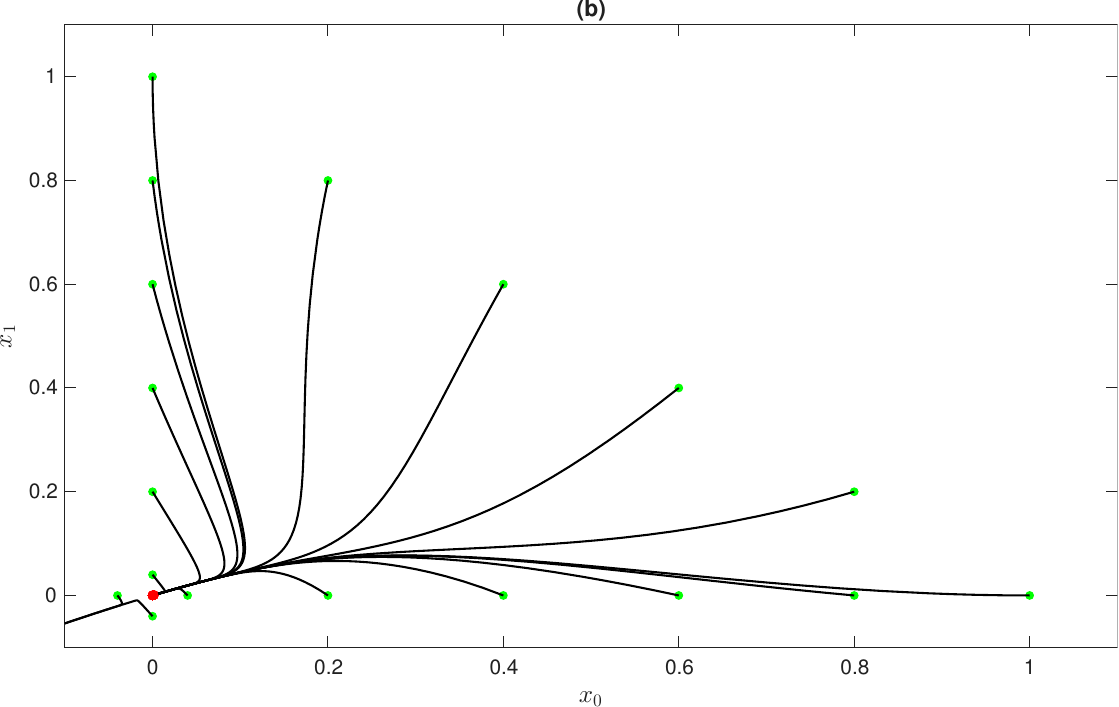}
\includegraphics[width=0.32\textwidth]{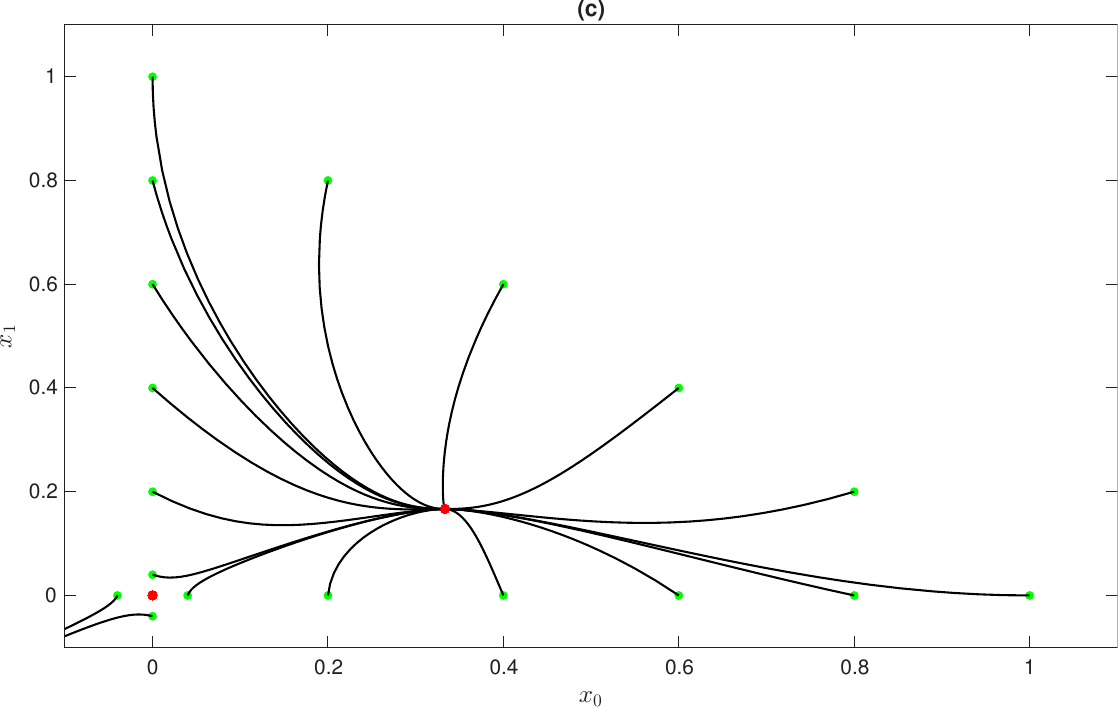}
    \caption{Trajectories of the deterministic model with $\mu_0=2$ 
and $\mu_1=4$,
and $q=\frac12$, and increasing values of $\lambda$:
(a) $\lambda=2$, (b) $\lambda=\mu^*=\frac83$, and (c)
$\lambda=\frac{16}{3}$.
The green dots indicate starting points.
The red dots indicate equilibria.}
    \label{fig:qequalshalf}
\end{figure}

Now assume $q\neq 1/2$.
For either root $r\in \{r^+,r^-\}$, the sum $m:=x_0+x_1$
is determined uniquely from either of (\ref{EQ1}) or (\ref{EQ2}):
$m=m(r)=1-\mu(r)/\lambda$, where
\begin{equation}
\mu(r) =\frac{\mu_0}{q+(1-q)r} =\frac{\mu_1 r}{1-q+qr}.
\label{muofr}
\end{equation}
Since $q\neq 1/2$, we may write
\begin{equation}
\mu(r^{\pm})
=\frac{2\mu_0}{q(1+s) \pm \sqrt{\delta}}
=\frac{\mu_0(-q(1+s) \pm \sqrt{\delta})}{2(1-2q)s}
=\frac{-q(\mu_0+\mu_1) \pm \mu_1\sqrt{\delta}}{2(1-2q)}.
\label{mustar}
\end{equation}
Also since $x_1=r x_0$, we have $x_0(1+r)=m$, and so 
$$
x_0(r) = \theta(r) m(r),
\qquad
x_1(r) = (1-\theta(r)) m(r),
$$
where $\theta(r)=1/(1+r)$.
This specifies two {\em distinct\/} equilibria, one for each root
$r\in \{r^+,r^-\}$. Notice that if $r=r^+$, then
$x_0$ and $x_1$ will be either both positive or both negative, 
while if $r=r^-$, then $x_0$ and~$x_1$ will have opposite signs,
the precise combination of signs being determined by the value of~$\lambda$. 
The only biologically relevant case has $r=r^+$
and $\lambda > \mu(r^+)$, implying that $\bx(r^+)\in E$. Let
us write $\mu^*=\mu(r^+)$ and $\theta^*=1/(1+r^+)$,
so that $\bx^*=(x_0^*,x_1^*)$ is given by
$$
x_0^* = \theta^* m^*, \quad x_1^* = (1-\theta^*) m^*,
\quad \text{and} \quad
m^*=1-\frac{\mu^*}{\lambda}.
$$
So, we may interpret $m^*$ as being the carrying capacity of the population,
which is apportioned to each trait according to the value of $r^+$.
Notice that on setting
$r=\mu_0/\mu_1$ ($=r^+$) in (\ref{muofr}) we get (\ref{EQmuconstant})
precisely. 
Note also that $f(1)=1-s>0$ (refer to (\ref{eq:r})) implies that
$r^+<1$, and hence that $x_1^*<x_0^*$,
reflecting the selective disadvantage of trait~1 ($\mu_1>\mu_0$).
We may think of $\mu^*$ as being an {\em effective mortality rate\/}. As we
shall see, it governs whether the population can be sustained.
It is interesting to note that $\mu_0<\mu^*<\mu_1$.
(This follows from (\ref{muofr}) with $r=r^+$, 
noting that $r^+<1$. First, $q+(1-q)r^+<q+1-q=1$, so $\mu^*>\mu_0$.
Next, $(1-q)(1-r^+)>0$ implies that $r^+<1-q+qr^+$, and so
$\mu^*<\mu_1$.)

Now set $\lambda_1=\mu(r^+)$ and $\lambda_2=\mu(r^-)$ for the two
crossover points (the values of $\lambda$ where $x_0$ and $x_1$ change signs).
We will see that $\lambda_1$ ($=\mu^*$) is indeed the critical value of 
$\lambda$ for stability of $\bzero$, but we will also need $\lambda_2$ for a
finer classification in the case of $q>1/2$.
In order to make it easier for us to 
distinguish the cases $q<1/2$ and $q>1/2$, let us write
$$
\lambda_1
=\frac{-q(\mu_0+\mu_1)+\mu_1\sqrt{\delta}}{2(1-2q)}
=\frac{q(\mu_0+\mu_1)-\mu_1\sqrt{\delta}}{2(2q-1)},
$$
and
$$
\lambda_2
=\frac{-q(\mu_0+\mu_1)-\mu_1\sqrt{\delta}}{2(1-2q)}
=\frac{q(\mu_0+\mu_1)+\mu_1\sqrt{\delta}}{2(2q-1)},
$$
and, importantly, note that $\lambda_1$ and $\lambda_2$ are
the roots of the quadratic equation
$$
g(\lambda)=0,
\quad \text{where }
g(\lambda)=(2q-1)\lambda^2 - (\mu_0 + \mu_1)q\lambda + \mu_0\mu_1.
$$
Observe that the graph of $g$ is a parabola with positive intercept;
it is downward opening if $q<1/2$, 
implying that $\lambda_2<0<\lambda_1$,
and upward opening if $q>1/2$, 
implying that $0<\lambda_1<\lambda_2$.

We can now proceed to classify the equilibria. At $\bzero$, we have
$$
B(\bzero)=
\left(
\begin{matrix}
\lambda q-\mu_0 & \lambda (1-q) \\
\lambda (1-q) &\lambda q-\mu_1  
\end{matrix}
\right).
$$
So, the trace of $B(\bzero)$ is equal to $2\lambda q -(\mu_0+\mu_1)$, 
and the determinant of $B(\bzero)$ is equal to $g(\lambda)$ (above).
And, since the discriminant 
$\Delta:=\tr^2-4d=4(1-q)^2\lambda^2 + (\mu_1-\mu_0)^2$ is strictly
positive, the eigenvalues of $B(\bzero)$ are real.
We get a different classification depending on the value of $q$
(which determines the shape of the graph of $g$).

Set $J=(\mu_0+\mu_1)/(2q)$. If $q<1/2$, then
(i) $\bzero$ is a stable node 
if $\lambda < J$ and $\lambda<\lambda_1$, 
(ii) it is a saddle point if $\lambda>\lambda_1$, and
(iii) it is an unstable node if
$\lambda > J$ and $\lambda<\lambda_2$ (which cannot occur).
If $q>1/2$, then
(i) $\bzero$ is a stable node 
if $\lambda < J$, and $\lambda<\lambda_1$
or $\lambda>\lambda_2$,
(ii) it is a saddle point if $\lambda_1<\lambda<\lambda_2$, and
(iii) it is an unstable node
if $\lambda > J$, and $\lambda<\lambda_1$ or $\lambda>\lambda_2$.
It can be shown that if $q<1/2$, then $\lambda_1<J$,
while if $q>1/2$, then $\lambda_1<J<\lambda_2$
(see Appendix~\ref{appendix6}),
so the classification becomes:
\begin{itemize}
\item[(i)] If $\lambda<\lambda_1$, then $\bzero$ is a stable node.
\item[(ii)] If $q<1/2$ and $\lambda>\lambda_1$, or if
$q>1/2$ and $\lambda_1<\lambda<\lambda_2$, then $\bzero$ is a saddle point.
\item[(iii)] If $q>1/2$ and $\lambda>\lambda_2$, 
then $\bzero$ is an unstable node.
\end{itemize}
Note that this classification encapsulates the case 
$\mu_0=\mu_1=\mu$ considered earlier, because on setting $s=1$ we get
$\sqrt{\delta}=2(1-q)$, and hence $\lambda_1=\mu$ and
$\lambda_2=\mu/(2q-1)$.

When $\bzero$ is unstable, small populations tend to grow regardless of
their trait composition. In contrast, when it is a saddle point, growth
depends on composition: some initial trait compositions lead to
persistence, while others lead to extinction. Thus, the saddle regime
reflects a threshold phenomenon in which successful establishment
requires not only favourable growth conditions, but also a favourable
balance of traits.

Next we classify the non-zero equilibria. At $\bx^*$ the Jacobian is
$$
B(\bx^*)=
\left(
\begin{matrix}
\lambda(q(1-2x_0^*)-x_1^*)-\mu_0 & \lambda ((1-q)(1-2x_1^*)-x_0^*) \\
\lambda((1-q)(1-2x_0^*)-x_1^*) & \lambda(q(1-2x_1^*)-x_0^*)-\mu_1  
\end{matrix}
\right).
$$
One finds that that its trace is
$$
\tr 
= \lambda(2q(1-m^*) -m^*) - (\mu_0+\mu_1)
= \mu^* + 2q(\mu^*- J) - \lambda.
$$
Also, its determinant is $d=\mu_1\sqrt{\delta}(\lambda-\mu^*)$ (the
proof, which is not completely straightforward, is given in
Appendix~\ref{appendix4}). Note that these formulae are valid when
$q=1/2$. Furthermore, the discriminant is a quadratic function of the
form $\Delta(\lambda)=\lambda^2 + b\lambda + c$, which is easily shown to
be strictly positive for $\lambda>0$.
In the biologically relevant case, $r=r^+$ and $\lambda>\lambda_1$, we 
have that $\tr = \lambda_1 +
2q(\lambda_1- J)-\lambda <\lambda_1-\lambda$, and
$d=\mu_1\sqrt{\delta}(\lambda-\lambda_1)$. So, the classification of
$\bx^*$ is very simple:
$\bx^*$ is a saddle point if $\lambda<\lambda_1$, and a stable node if
$\lambda>\lambda_1$ (it cannot be an unstable node).

\section{Quasi stationarity}

While the deterministic approximation captures overall trends and
long-term equilibria, the stochastic model exhibits persistent
fluctuations around these states. In regimes where the population
persists, the process typically remains close to a stable equilibrium
for long periods before eventual extinction. Understanding this
quasi-equilibrium, or quasi-stationary, 
behaviour is therefore essential for describing the
observable dynamics of the system.

The law of large numbers (Theorem~\ref{LLN} of Appendix~\ref{appendix2})
gives a precise statement about convergence of sample paths of our model
to the corresponding deterministic trajectories. It does not say
anything about the random fluctuations about those paths, fluctuations
which are apparent in all of the simulations illustrated in 
Figure~\ref{fig:range}. The central limit law
(Theorem~\ref{CLT} of Appendix~\ref{appendix3}) shows that these
fluctuations follow a diffusion process whose parameters can be
determined from the model parameters. However, 
in cases where the process reaches a quasi equilibrium, this regime
is reached rapidly (a behaviour that can be seen most clearly in
Figure~\ref{fig:range}b), the {\OU} process described in
Corollary~\ref{CLTOU} of Appendix~\ref{appendix3} will be of greater use
in describing this behaviour, especially as we have explicit expressions
available for its stationary covariance matrix. In particular, when
$N$ is large,
$\bX=\bn/N$ has an approximate bivariate normal distribution with
$\Var(X_0)$, $\Var(X_1)$, and $\Cov(X_0,X_1)$, approximated by
$\Var_0/N$, $\Var_1/N$, and $\Cov_{01}/N$, respectively, where 
$\Var_0$, $\Var_1$, and $\Cov_{01}$, are given by 
(\ref{Var0}), (\ref{Var1}), and (\ref{Cov}).

Figure~\ref{fig:contourplot} shows the simulation of
$\bX(t)=(X_0(t),X_1(t))$ 
corresponding to the sample paths $X_0$ and $X_1$ depicted
individually in Figure~\ref{fig:range}b, together with a contour
plot of the approximating bivariate normal distribution.
We see that once the process reaches quasi equilibrium, the
approximation is faithful (the outer contour marks $98$ percent of the 
probability mass).

In the no-selection case, $\mu_0=\mu_1=\mu$, with $\lambda>\mu$ for
a positive equilibrium $\bx^*=(\rho/2,\rho/2)$
($\rho=1-\mu/\lambda$), the covariance matrix can be evaluated explicitly:
$$
\Var_0=\Var_1 =\frac{\lambda - \mu(2q-1)}{8\lambda(1-q)} \ (>0),
\qquad 
\Cov_{01} = -\frac{\lambda - \mu(3 - 2q)}{8\lambda(1-q)}.
$$
We see that $\Cov_{01}<0$, $\Cov_{01}=0$, or $\Cov_{01}>0$, according to
whether $\lambda>(3-2q)\mu$, $\lambda=(3-2q)\mu$, or
$\lambda<(3-2q)\mu$ (note that $3-2q>1$). 
These conditions imply, respectively, that when $N$ is large,
the trait numbers, $n_0$ and $n_1$,
are asymptotically negatively correlated, 
uncorrelated, or positively correlated.
We can make some headway in the general case noticing that
$\eta$ in (\ref{eta}) is equal to
$\tr(B(\bx^*))\det(B(\bx^*)$, which is strictly negative since $\bx^*$
is a stable node. It follows that
$$
\sgn\{\Cov_{01}\}=-\sgn\{b_{11}b_{12}\mu_1x_1^*+b_{21}b_{22}\mu_0x_0^*\},
$$
which may simplify identifying the sign change. Note
that this sign change can be understood by considering the
fluctuations of the total population and the imbalance between traits.
Set $M = X_0 + X_1$ and $D = X_1 - X_0$ 
(following~Section~\ref{section:noselection}). Then,
\begin{equation}
\mathrm{Cov}(X_0,X_1) = \tfrac{1}{4}\bigl(\mathrm{Var}(M)
- \mathrm{Var}(D)\bigr).
\label{CC}
\end{equation}
Thus, the sign of the covariance reflects a competition between two  
sources of variability: fluctuations in total
population size and fluctuations in the imbalance between traits. When
variation in the total population dominates, the traits tend to co-vary
positively; when fluctuations in their difference dominate, the traits
become negatively correlated. This provides a simple interpretation of
stochastic variability in terms of competing ecological and
compositional fluctuations.

\begin{figure}[htbp]
    \centering
    \includegraphics[width=0.8\textwidth]{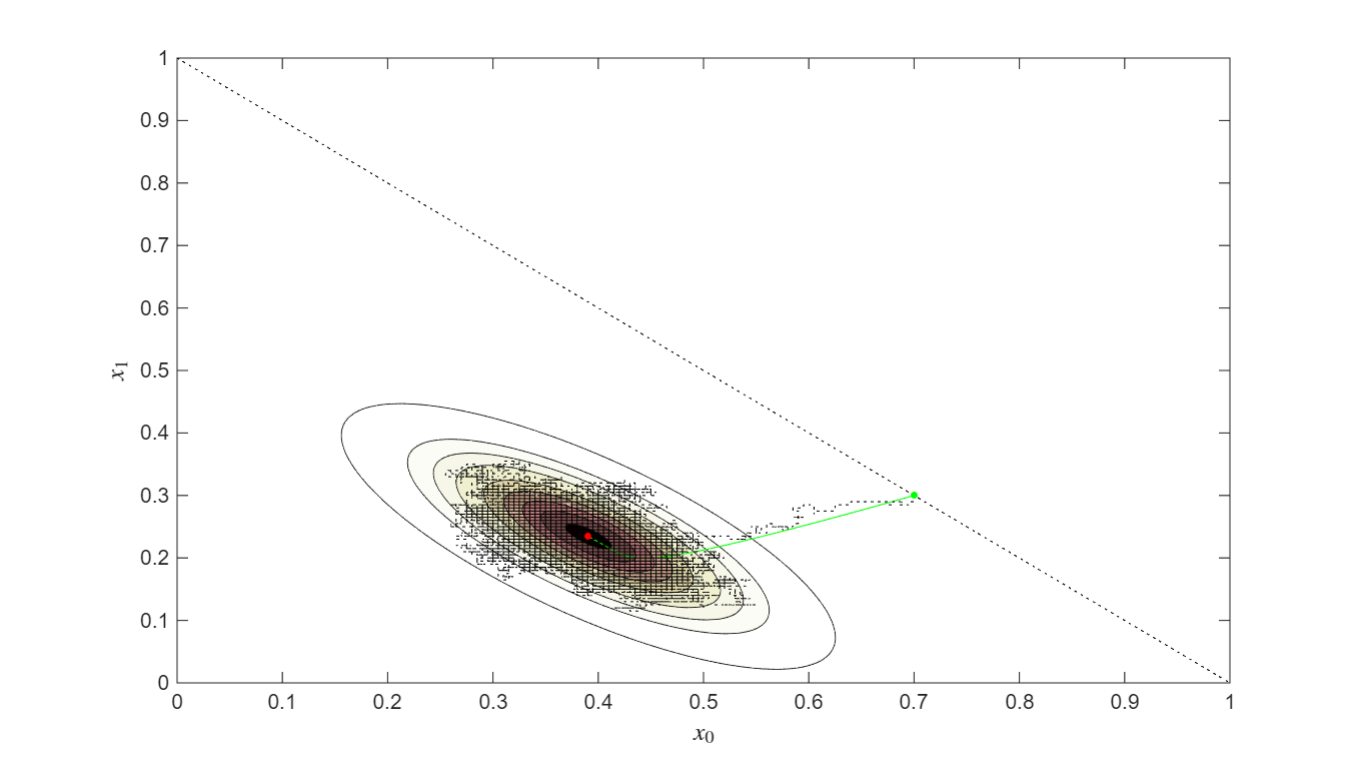}
    \caption{Simulation of the process $\bX(t)=(X_1(t),X_2(t))$ 
with $N=200$, $\lambda = 2.5$, 
$q = 0.9$, $\mu_0 = 0.9$, $\mu_1 = 1$ over the time
interval $(0,60)$. The approximating deterministic trajectory
is also shown, together with a contour plot of the approximating 
discretised bivariate normal distribution.}
    \label{fig:contourplot}
\end{figure}

\subsection{The {\qsd}}

The notion of a {\qsd} is a very natural one. Think of an observer who
at some time $t$ is aware that the population is extant, yet cannot
discern the trait composition. What is the chance of there being a
particular trait composition $\bn=(n_0,n_1)$? If we were equipped with
the complete set of state probabilities $p_{\bn}(t)=\Pr(\bn(t) =\bn)$,
$\bn\in E$, we would evaluate the conditional probability
$$
u_{\bn}(t)=\Pr(\bn(t) =\bn | \bn(t)\neq \bzero)
= \frac{p_{\bn}(t)}{1-p_{\bzero}(t)},
\qquad
\bn\in S^{\myprime}:=E\backslash \{\bzero\},
$$
because we have observed the event $\{\bn(t) \neq \bzero\}$.
Then, in view of the quasi-equilibrium behaviour observed earlier, it
would be natural to seek a distribution $\bu=(u_{\bn},\, \bn\in
S^{\myprime})$ over~$S^{\myprime}$ such that if $u_{\bn}(s)=\bu$ for a
particular~$s$, then $u_{\bn}(t)=\bu$ for all $t>s$. Such a distribution
is called a {\em stationary conditional distribution\/} or {\em
quasi-stationary distribution\/}.

Since $S^{\myprime}$ is finite and irreducible, the {\qsd}~$\bu$ is
unique, being the eigenvector of the $q$-matrix $Q$ restricted
to~$S^{\myprime}$, corresponding to the eigenvalue $-\nu$ with maximum
real part (the $q$-matrix is the matrix of transition rates with
diagonal entries $q(\bn,\bn)$ equal to $-q(\bn)$, where $q(\bn)$, the
total rate out of state $\bn$, is given by~(\ref{qn})): $Q\bu =-\nu
\bu$. The eigenvalue $-\nu$ is real, simple, and strictly negative.
Furthermore, again since $S^{\myprime}$ is finite and irreducible, $\bu$
is also a {\em limiting conditional distribution\/} in that
$u_{\bn}(t)\to \bu$ as $t\to\infty$, whereever the process is first
observed. The quantity $\nu$ is often called the {\em decay
parameter\/}. For example, if $T$ is the time of extinction, then
$\Pr_{\bu}(T>t)=e^{-\nu t}$. That is, if $\bu$ is the
initial distribution over states (we start in quasi equilibrium),
then the time to extinction is exponentially distributed with mean $1/\nu$.
For further details, see~\cite{vDP12}.

Thus, in principal, we can evaluate~$\bu$ for our model. However, in all
but a handful of simple $1$-dimensional Markov chains, $\bu$ cannot be
exhibited explicitly. Even for the continuous-time SIS model mentioned
earlier, which in our case governs the behaviour of $s(t)=n_0(t)+n_1(t)$
when there is no selection ($\mu_0=\mu_1$), there is no explicit
expression available for $\bu$. Furthermore, standard numerical methods
break down if the state space is even moderately large. For example, if
$N=1000$, then the restricted $Q$ has $501\,500^2$
($=251\,503\,253\,001$) entries, and thus would require approximately
$2.01$TB of RAM to be stored. For this reason, approximation methods are
preferred. The {\OU} approximation has particular appeal as its
stationary distribution is expected to be close the {\qsd} once suitably
discretised. (By ``discretised'' we mean that, for all $\bn=(n_0,n_1)$,
we evaluate the probability that our bivariate normal random variable
lies in the square delimited by points $(n_0 \pm 0.5)/N$ and $(n_1 \pm
0.5)/N$). We will illustrate this closeness by evaluating $\bu$ for our
model with parameters $N=200$, $\lambda = 2.5$, $\mu_0 = 0.9$, $\mu_1 =
1$, and $q = 0.9$ (the same values used in the simulations depicted in
Figures~\ref{fig:range}b and~\ref{fig:contourplot}). For these parameter
values only $0.024$ percent of the entries in the restricted $Q$ are
non-zero, so it is convenient to use the sparse-matrix methods outlined
in Appendix~A of~\cite{vDP12} (supplementary material). The
quasi-stationary distribution was evaluated using Matlab's {\tt eig}
command with the default settings. Figure~\ref{fig:qsdversusOU}
illustrates the result. The dotted lines are the quasi-stationary
marginal distributions of $X_0$ (red) and $X_1$ (blue), while the solid lines
correspond to probabilities predicted by the {\OU} approximation. The means
of these distributions are located near $x_0^*=0.39063$ and
$x_1^*=0.23437$, respectively. The corresponding approximating
standard deviations of $X_0$ and $X_1$ are, respectively, $0.065306$ and
$0.059152$, which appear to be consistent with sizes of the fluctuations
depicted in Figure~\ref{fig:range}b. It is interesting to note that the
decay parameter $\nu$ was estimated to be $5.66\times 10^{-14}$, which
indicates that the quasi stationary regime is very long-lived.

The above comparison supports the use of the {\OU} approximation as an
effective approximation to the quasi-stationary distribution.

\begin{figure}[htbp]
    \centering
    \includegraphics[width=0.8\textwidth]{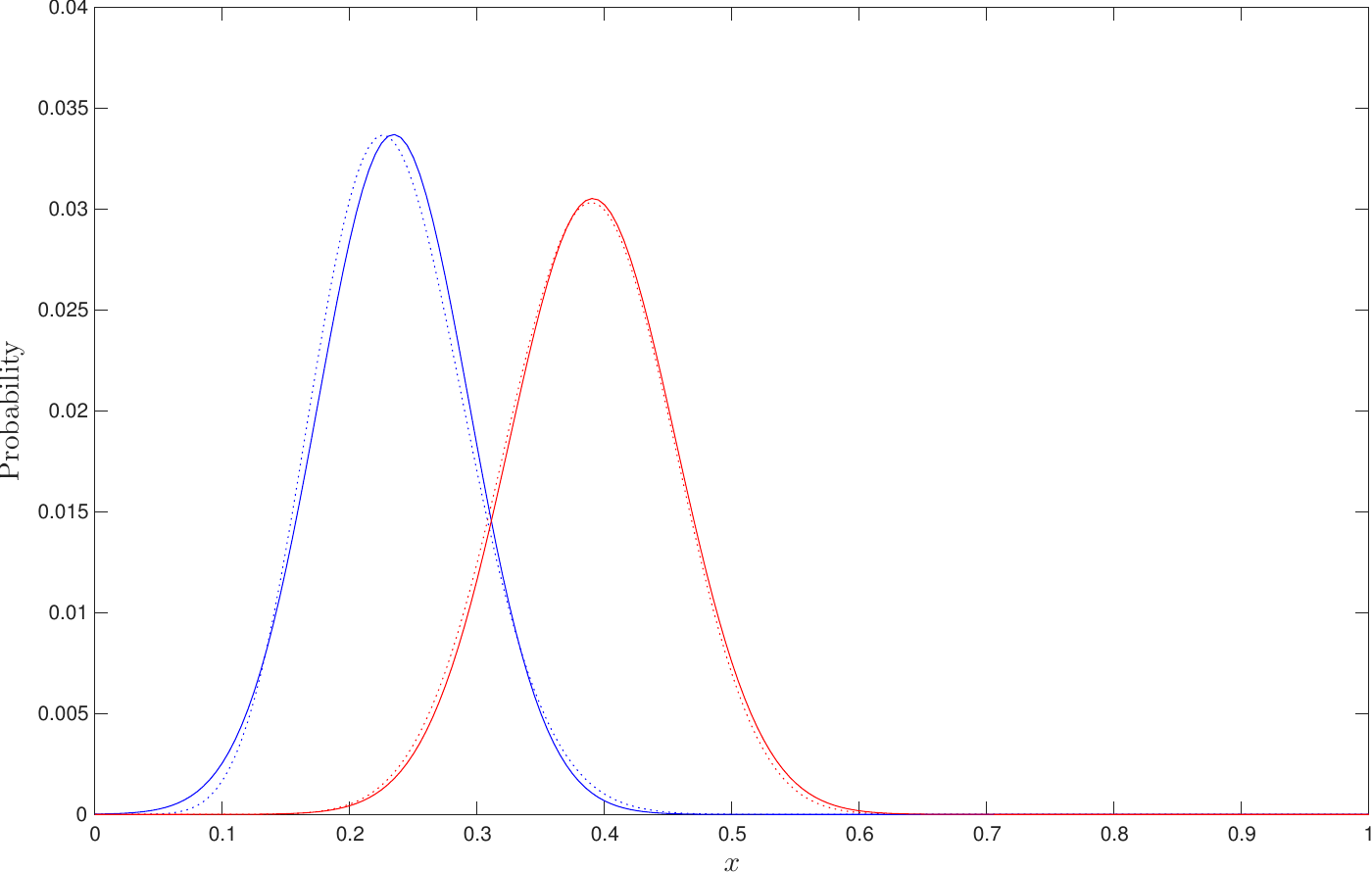}
    \caption{A comparison between the {\qsd} and the {\OU}
approximation. The dotted lines are the quasi-stationary marginal
distributions of $X_0$ (red) and $X_1$ (blue). The solid lines
correspond to probabilities predicted by the {\OU} approximation. Here
$N=200$, $\lambda = 2.5$, $\mu_0 = 0.9$, $\mu_1 = 1$, and $q = 0.9$.}
    \label{fig:qsdversusOU}
\end{figure}

\section{Conclusion}

We have examined how mutation, selection, and density dependence
interact to determine both population persistence and trait composition.
By embedding imperfect inheritance within a regulated population, the
model shows that mutation does more than introduce variation: it
modifies the effective demographic balance. In particular, mutation
induces an effective mortality rate that governs whether the population
can be sustained, providing a simple criterion for persistence. Beyond
this primary threshold, the model reveals a second threshold in which
population establishment depends on initial trait composition. When
inheritance dominates mutation, some initial configurations lead to
growth while others lead to extinction, even when overall growth rates
are favourable. This composition-dependent behaviour represents a
qualitative departure from classical models in which persistence is
determined solely by net reproduction.

The analysis also clarifies the role of mutation in maintaining
diversity. Although one trait is selectively inferior, it can persist
through continual reintroduction via mutation, leading to a non-trivial
mutation-selection balance. In stochastic regimes, the population
exhibits long-lived quasi-equilibrium behaviour prior to extinction,
well captured by a diffusion approximation. The stochastic analysis
further reveals a transition in the sign of trait correlations,
reflecting a shift between fluctuations in total population size and in
trait imbalance.

These findings arise within a deliberately simple two-trait framework. While this setting allows a detailed and largely explicit analysis, it also imposes limitations, including homogeneous mixing, constant environmental conditions, and a simple mutation mechanism governed by a single parameter. Selection acts only through mortality, and density dependence is imposed through a fixed population ceiling.

Several natural extensions suggest themselves. Extending the model to multiple or continuous traits would provide a richer description of mutation–selection dynamics under ecological regulation. Allowing time-varying mutation rates or environmental conditions could reveal additional interactions between variability and persistence thresholds. It would also be of interest to consider alternative forms of density dependence arising from explicit resource dynamics, and to explore how these affect the effective mortality concept and associated threshold behaviour.

Overall, the results highlight how even simple forms of imperfect
inheritance can qualitatively reshape both ecological persistence and
evolutionary outcomes. In regulated populations, mutation is not merely
a source of variation, but a structural component of the dynamics,
governing both persistence and the maintenance of trait diversity.

\section{AI usage}

The author used Microsoft Copilot (GPT-based AI) as a tool to assist
with phrasing, exposition, and checking intermediate calculations. The
author takes full responsibility for all content, including all
mathematical results.

\section{Appendix}

\subsection{Existence and uniqueness}
\label{appendix1}

Here we establish that there is a unique solution to~(\ref{ODE1}) in~$E$.

\begin{lemma}\label{locallyLipschitz}
$F$ given by (\ref{F0x}) and (\ref{F1x}) is locally Lipschitz in $\R^2$.
\end{lemma}

\begin{proof}
Write $F(\bx) =G(\bx) - D\bx$, where $G(\bx)=\lambda (1-m(\bx))P\bx$,
$D=\operatorname{diag}(\mu_0,\mu_1)$, and
$$
P=
\begin{pmatrix}
q   & 1-q \\
1-q & q
\end{pmatrix}.
$$
Then, for all $\bx,\by\in \R^2$,
$$
\norm{D(\bx-\by)}_1
=\sum_i |(D\bx)_i - (D\by)_i|
=\sum_i |\mu_i x_i - \mu_i y_i|
\leq \mu_1 \sum_i |x_i - y_i|
=\mu_1 \norm{\bx-\by}_1.
$$
Also, for all $\bx,\by\in \R^2$,
$$
G(\bx)-G(\by) = \lambda \left( (1-m(\bx)) P(\bx-\by) +
(m(\by)-m(\bx))P\by \right),
$$
and so, since 
$|m(\bx)-m(\by)| \leq \sum_i |x_i -y_i| =\norm{\bx-\by}_1$, 
and $\norm{P\bx}_1 \leq \norm{\bx}_1$, we have
\begin{align*}
\norm{G(\bx)-G(\by)}_1 
&\leq  \lambda 
\left( |1-m(\bx)| \norm{\bx-\by}_1 + |m(\by)-m(\bx)|\norm{\by}_1
\right)\\
&\leq
\lambda (1 + \norm{\bx}_1 + \norm{\by}_1)\norm{\bx-\by}_1.
\end{align*}
Therefore, $\norm{G(\bx)-G(\by)}_1 \leq \lambda (1 + 2R)\norm{\bx-\by}_1$ on
the ball $B_R=\{\bx\in \R^2: \norm{\bx}_1\leq R\}$, and
so $\norm{F(\bx)-F(\by)}_1 
\leq (\mu_1+\lambda (1 + 2R))\norm{\bx-\by}_1$ on $B_R$.
\end{proof}

Lemma~\ref{locallyLipschitz} allows us to apply the 
Picard–Lindel{\"o}f Theorem (see for example \cite{VE21}): 
if $U$ is an open subset of~$\R^2$, then
there is an $\epsilon>0$ such that, if $\bx(0)\in U$, then $\bx(t)\in U$ 
for all $t\in (-\epsilon,\epsilon)$.
However, we will show that any solution to~(\ref{ODE1}) 
which starts in $E$, remains in $E$.
(It is obvious that the interior of $E$ is attracting, because
on the boundaries of $E\backslash \{\bzero\}$, 
the vector field points inwards: 
$F_0(0,x_1)=\lambda(1-q)x_1(1-x_1)>0$, for $0<x_1<1$,
$F_1(x_0,0)=\lambda(1-q)x_0(1-x_0)>0$, for $0<x_0<1$,
and on the upper boundary, $m=1$, we have 
$\dot{m}=-\mu_0 x_0 -\mu_1 x_1 <0$.)

\begin{theorem}
If $\bx_0\in E$, then $\bx(t)\in E$ for all $t>0$.
\end{theorem}

\begin{proof}
Integrating $\dot{\bx}=G(\bx) - D\bx$ using the integrating factor
$e^{Dt}$ we get
$$
\bx(t)=e^{-D t}\bx(0)+\int_0^t e^{-D(t-s)}\,G(\bx(s))\,ds,
$$
or, componentwise,
\begin{equation}
x_i(t)=e^{-\mu_i t}x_i(0)+\int_0^t
e^{-\mu_i(t-s)}\,\lambda\bigl(1-m(\bx(s))\bigr)\, (P\bx(s))_i\,ds,
\qquad i\in \{0,1\}.
\label{mildsolution1}
\end{equation}

Suppose that $\bx(0)\in E$, that is, 
$x_i(0)\geq 0$ for both $i$, and $m(\bx(0))\leq 1$.
We will first show $m(\bx(t))\leq 1$ for all $t$.
Let $t_*:=\inf\{t>0:\ m(\bx(t))>1\}$ 
($t_*=\infty$ if the set is empty).
Fix $t<t_*$. Then $m(\bx(s))\le 1$ for all $s\in[0,t]$, and hence 
$1-m(\bx(s))\geq 0$ on this interval. 
On summing (\ref{mildsolution1}) over $i$ we find that 
\[
m(\bx(t))=\sum_i e^{-\mu_i t}x_i(0)
+\int_0^t \lambda(1-m(\bx(s)))\sum_i e^{-\mu_i(t-s)}(P\bx(s))_i\,ds.
\]
But, $\sum_i e^{-\mu_i t}x_i(0)\le \sum_i x_i(0)=m(\bx(0))$, and
$$
\sum_i e^{-\mu_i(t-s)}(P\bx(s))_i
=\sum_k x_k(s)\sum_i e^{-\mu_i(t-s)}p_{i-k}\\
\leq \sum_k x_k(s)=m(\bx(s)),
$$
and hence 
\begin{equation}
m(\bx(t))\le m(\bx(0))+\int_0^t \lambda(1-m(\bx(s)))\,m(\bx(s))\,ds.
\label{mbxt1}
\end{equation}
Now compare the right-hand side of~(\ref{mbxt1}) with the
integrated form of the solution to the
logistic equation $\dot{u}=f(u)$, where $f(u)=\lambda u(1-u)$ and 
$u(0)=u_0:=m(\bx(0))\le 1$:
$$
u(t) = u_0 + \int_0^t \lambda (1-u(s)) u(s)\, ds.
$$
We will show that $m(\bx(t))\leq u(t)$ for all $t\geq 0$. 
Since $u(t)<1$ for all~$t$, it will follow that $t_*=\infty$. 
Let $\phi_t(x_0)$ be the flow corresponding to $\dot{u}=f(u)$,
and notice that the flow is monotone in that if $u_2\geq u_1>0$,
then $\phi_t(u_2)\geq \phi_t(u_1)$ for all~$t$.
Let $\tau := \inf\{t>0 : m(\bx(t)) > u(t)\}$.
Then, $m(\bx(t))\leq u(t)$ for all $t<\tau$, and $m(\bx(\tau))=u(\tau)$
by continuity. Now fix $h>0$ and observe that
$$
m(\bx(\tau+h)) \leq \phi_h(m(\bx(\tau)))=\phi_h(u(\tau))= u(\tau+h),
$$
which is a contradiction. Therefore $m(\bx(t))\leq u(t)$.
It is now easy to see from (\ref{mildsolution1}) that $x_i(t)\geq 0$
for all $i$, since the integrand is non-negative.
So, finally, $\bx(t)\in E$ for all $t>0$.
\end{proof}

\subsection{Law of large numbers}
\label{appendix2}

We will apply Theorem~3.1 of~\cite{Kur70}. Using Kurtz's notation,
$$
f(\bx,\bl) =
\begin{cases}
\lambda(q x_0 + (1-q)x_1)(1-x_0-x_1) &\text{if } \bl=(1,0) \\
\lambda((1-q) x_0 + qx_1)(1-x_0-x_1) &\text{if } \bl=(0,1) \\
\mu_0 x_0 &\text{if } \bl=(-1,0) \\
\mu_1 x_1 &\text{if } \bl=(0,-1).
\end{cases}
$$
Then, $F(\bx)=\sum_{\bl} \bl f(\bx,\bl)$, as given by
(\ref{F0x}) and (\ref{F1x}).
Referring to the conditions of that theorem,
it is readily seen from the proof of Lemma~\ref{locallyLipschitz}
that $F$ is Lipschitz continuous on $E$ with Lipschitz constant
$\mu_1+3\lambda$.
Next,
$$
\sup_{x\in E} \sum_{\bl} |\bl| f(\bx,\bl)
=
\sup_{x\in E} \left( (x_0+x_1)(1-x_0-x_1) + \mu_0 x_0 + \mu_1 x_1
\right) \leq 1+ \mu_1 \ (<\infty),
$$
and so Condition~(3.3) holds good. 
For $d>1$, there is no $\bl$ such
that $|\bl|>d$, and so Condition~(3.4), that
$$
\lim_{d\to\infty} \sup_{x\in E} \sum_{|\bl|>d} |\bl| f(\bx,\bl) =0,
$$
is trivially satisfied.
Define a family of Markov chains $\bXN=(\bXN(t),\, t\geq 0)$,
where $\bXN(t)=\bn(t)/N$. Then, Theorem~3.1 of~\cite{Kur70} allows
us to conclude the following.

\begin{theorem}
\label{LLN}
If $\lim_{N \to \infty} \bXN(0)=\bx_0$,
for some $\bx_0\in E$, then
$\bXN$ converges in probability, uniformly
over finite time intervals, to $(\bx(t),\, t\geq 0)$,
the unique solution to (\ref{ODE1}) in $E$ with prescribed initial
condition $\bx(0)=\bx_0$, that is
$$
\lim_{N \to \infty}
\Pr\left(\sup_{0\leq s\leq t} \left| \bXN(s)-\bx(s)\right|>\epsilon\right)=0,
$$
for all $t>0$ and for all $\epsilon>0$.
\end{theorem}

\noindent
Convergence over {\em finite\/} time intervals is very important here,
because, as we have already observed, $\bn(t)$, hence $\bXN(t)$, reaches
$\bzero$ in finite mean time.

\subsection{Diffusion approximation}
\label{appendix3}

It will be convenient to apply Theorem~3.2 of Pollett~\cite{Pol90},
which is derived from Theorems~3.1 and~3.5 of Kurtz~\cite{Kur71}. 
Let $G(\bx)$ be the ($2\times 2$) matrix defined by
$G_{ij}(\bx) =\sum_{\bl} l_i l_j f(\bx,\bl)$.
So, the off diagonal entries are $0$, and diagonal entries are
$$
G_{00}(\bx) = \lambda(q x_0 + (1-q)x_1)(1-x_0-x_1)+\mu_0 x_0
= F_0(\bx)+2\mu_0 x_0,
$$
$$
G_{11}(\bx) = \lambda((1-q) x_0 + q x_1)(1-x_0-x_1)+\mu_1 x_1
= F_1(\bx)+2\mu_1 x_1.
$$
In addition to the conditions checked in Appendix~\ref{appendix2},
we need to check that $F$ has uniformly continuous first partial
derivatives. This is clearly true because
$\nabla F(\bx)$ (Eqn.~\ref{Jac}) is a quadratic function. 
So, we have the following central limit law.

\begin{theorem}
\label{CLT}
If $\lim_{N \to \infty} N^{\frac12} \left(\bXN(0)-\bx(0) \right)=\bz$,
then the family of processes $(\bZN(t),\, t\geq 0)$, defined by
$$
\bZN(s)=N^{\frac12} \left(\bXN(s)-\bx(s) \right),
\qquad 0\leq s \leq t,
$$
converges weakly in $D[0,t]$
(the space of right-continuous, left-hand limits functions on $[0,t]$)
to a Gaussian diffusion $(\bZ(t),\, t\geq 0)$, with initial value $\bZ(0)=\bz$,
and with $\E(\bZ(s))=M_s \bz$, where
$M_s=\exp \left( \int_0^s B_u du \right)$, $B_s=\nabla F(\bx(s))$.
\end{theorem}

\noindent
The covariance matrix $\Sigma_s$ of $\bZ(s)$ is the unique solution to
$$
\frac{d\Sigma_s}{ds} = B_s \Sigma_s  +\Sigma_s  B_s ^T+G(X(s)),
$$
that is,
$$
\Sigma_s=M_s\left( \int_0^s M_u^{-1} G(\bx(u)) (M_u^{-1})^T du \right) M_s^T.
$$
If $\bXN$ starts near an equilibrium point $\bx^*$
of the limiting deterministic model, 
then we can be more precise about the limiting diffusion.

\begin{corollary}
\label{CLTOU}
Suppose $\lim_{N \to \infty} \bXN(0)=\bx^*$. Then, if
$\lim_{N \to \infty} N^{\frac12} \left(\bXN(0)-\bx^* \right)=\bz$,
the family of processes $(\bZN(t),\, t\geq 0)$, defined by
$$
\bZN(s)=N^{\frac12} \left(\bXN(s)-\bx^* \right), \qquad 0\leq s\leq t,
$$
converges weakly in $D[0,t]$ to an {\OU} process $(\bZ(t),\, t\geq 0)$
with local drift matrix $B=\nabla F(\bx^*)$, local covariance matrix
$G=G(\bx^*)$, and with initial value $\bZ(0)=\bz$.
\end{corollary}

\noindent
In particular, $\bZ(s)$ is normally distributed with mean
$M_s=e^{Bs} \bz$ and covariance matrix
$$ 
\Sigma_s =\int_0^s e^{Bu} G e^{B^T u} du =\Sigma-e^{Bs} \Sigma e^{B^T s},
$$
where $\Sigma$, the stationary covariance matrix, satisfies
$$
B\,\Sigma +\Sigma B^T+G=0.
$$
Note that $\Cov(\bZN(s))=N\Cov(\bXN(s))$, so we may approximate
$\Cov(\bXN(s))$ by $\Sigma/N$.
We can solve $B\Sigma +\Sigma B^T+G=0$ for our model, because
we have, simply,
$$
G_{ii}(\bx^*) = F_i(\bx^*)+2\mu_i x_i^*=2\mu_i x_i^*,
$$
and so $G$ is the diagonal matrix with diagonal entries $2\mu_i x_i^*$.
Writing 
$$
\Sigma=
\left(
\begin{matrix}
{\Var}_0 & \Cov_{01} \\
\Cov_{01}     & {\Var}_1
\end{matrix}
\right),
\quad
\text{and}
\quad
B:=\nabla F(x_0^*,x_1^*)=
\left(
\begin{matrix}
b_{11} &b_{12} \\
b_{21} &b_{22}
\end{matrix}
\right),
$$
we get
\begin{equation}
{\Var}_0=
((b_{12}b_{21}-b_{11}b_{22}-b_{22}^2)\mu_0x_0^*-b_{12}^2\mu_1x_1^*)/\eta,
\label{Var0}
\end{equation}
\begin{equation}
{\Var}_1
=((b_{12}b_{21}-b_{11}b_{22}-b_{11}^2)\mu_1x_1^*-b_{21}^2\mu_0x_0^*)/\eta,
\label{Var1}
\end{equation}
and
\begin{equation}
\Cov_{01}
=(b_{11}b_{12}\mu_1x_1^* + b_{21}b_{22}\mu_0x_0^*)/\eta,
\label{Cov}
\end{equation}
where
\begin{equation}
\eta=b_{11}^2b_{22}-b_{11}b_{12}b_{21}+b_{11}b_{22}^2-b_{12}b_{21}b_{22}.
\label{eta}
\end{equation}

\subsection{Proof that $d=\pm\mu_1\sqrt{\delta}(\lambda-\mu^*(r^{\pm}))$}
\label{appendix4}

On writing out the determinant of $B(\bx^*)$ 
we find that, for $r=r^{\pm}$,
$$
d=d(r)=K(r)(\lambda-\mu^*(r)) + g(\mu^*(r)),
$$
where
\begin{equation}
K(r)=(1-2q)\mu^*(r) 
+ \frac{(\mu_1-\mu_0)q(1-r) + \mu_1 r + \mu_0}{1+r}.
\label{Kr}
\end{equation}
But, $g(\mu^*(r))=0$, and
on substituting for $\mu^*(r)$ (eqn.~(\ref{mustar})) and $r$
(eqn.~(\ref{rpm})) in (\ref{Kr})
and using the expression for $\delta$ (eqn.~(\ref{delta})), we find that
$$
K(r)=\frac{\mu_1^2\delta \pm A\mu_1\sqrt{\delta}}{A \pm \mu_1\sqrt{\delta}},
$$
where $A=(\mu_0-3\mu_1)q + 2\mu_1$. On multiplying the denominator and
the numerator by $A \mp \mu_1\sqrt{\delta}$, we arrive at
$K(r)=\pm \mu_1 \sqrt{\delta}$ ($A$ disappears).

\subsection{Proof that $h$ does not depend on $\bx^*$}
\label{appendix5}

Differentiate $F(\bx)=A\bx + h(\bx)$, where $A=\nabla F(\bzero)$, to get
$\nabla F(\bx)=A + \nabla h(\bx)$.
Now supppose that there is a function~$g$ (necessarily quadratic)
such that $F(\bx)=B(\bx-\bx^*) + g(\bx-\bx^*)$,
where $B= \nabla F(\bx^*)$. Differentiate this to obtain
$\nabla F(\bx)=B + \nabla (g(\bx-\bx^*)) = B+\nabla g(\bx-\bx^*)$.  So,
$A + \nabla h(\bx) =B + \nabla g(\bx-\bx^*)$.
Let $\by=\bx-\bx^*$. Then,
$$
A + \nabla h(\by+\bx^*) =B + \nabla g(\by).
$$
Since $\nabla h$ is linear,
$\nabla h(\by+\bx^*) = \nabla h(\by) + \nabla h(\bx^*)$. Therefore,
$$
A + \nabla h(\by) + \nabla h(\bx^*) =B + \nabla g(\by).
$$
But $A + \nabla h(\bx^*) = \nabla F(\bx^*)=B$, and so
$\nabla h(\by) = \nabla g(\by)$. Thus,
$h(\by)$ and $g(\by)$ differ by a constant. This constant must be
$\bzero$
because $h(\bzero)=F(\bzero)=0$ and $g(\bzero)=F(\bx^*)=0$.
We deduce that $g=h$.

\subsection{The ordering of $\lambda_1$, $\lambda_2$ and $J$}
\label{appendix6}

Assume $q<1$, and recall that
$$
\lambda_1
=\frac{-q(\mu_0+\mu_1)+\sqrt{\delta}}{2(1-2q)}
=\frac{q(\mu_0+\mu_1)-\sqrt{\delta}}{2(2q-1)},
\quad
\lambda_2
=\frac{-q(\mu_0+\mu_1)-\sqrt{\delta}}{2(1-2q)}
=\frac{q(\mu_0+\mu_1)+\sqrt{\delta}}{2(2q-1)},
$$
where
$\delta= (\mu_0 + \mu_1)^2 q^2 - 4\mu_0\mu_1 (2q-1)$ ($>0$).
We first prove that if $q<1/2$, then $\lambda_1<J:=(\mu_0+\mu_1)/(2q)$.
Assuming the contrary, we have a series of implications:
\begin{align*}
\lambda_1\geq J 
&\implies q(-q(\mu_0+\mu_1)+\sqrt{\delta})\geq (1-2q)(\mu_0+\mu_1)\\
&\implies q\sqrt{\delta} \geq (1-q)^2 (\mu_0+\mu_1)\\
&\implies (\mu_0+\mu_1)^2 q^4 + 4\mu_0\mu_1 q^2(1-2q) \geq  (1-q)^4
(\mu_0+\mu_1)^2\\
&\implies 4\mu_0\mu_1 q^2(1-2q) \geq (1-2q)(1-2q+2q^2)(\mu_0+\mu_1)^2\\
&\implies 2q^2
\left( \frac{2\mu_0\mu_1}{(\mu_0+\mu_1)^2}-1\right)\geq 1-2q \ (>0).
\end{align*}
But this cannot be, because
$2\mu_0\mu_1/(\mu_0+\mu_1)^2 \leq 1/2$ 
(harmonic mean $\leq$ arithmetic mean),
and hence the left-hand side is less than or equal to $-q^2<0$.

Next we will prove that if $q>1/2$, then $\lambda_1<J$.
Assuming the contrary, we see that
\begin{align*}
\lambda_1\geq J 
&\implies q(q(\mu_0+\mu_1)-\sqrt{\delta})\geq (2q-1)(\mu_0+\mu_1)\\
&\implies q\sqrt{\delta} \leq (1-q)^2 (\mu_0+\mu_1)\\
&\implies (\mu_0+\mu_1)^2 q^4 - 4\mu_0\mu_1 q^2(2q-1) \leq  (1-q)^4
(\mu_0+\mu_1)^2\\
&\implies 4\mu_0\mu_1 q^2(2q-1) \geq (2q-1)(1-2q+2q^2)(\mu_0+\mu_1)^2\\
&\implies 2q^2
\left( \frac{2\mu_0\mu_1}{(\mu_0+\mu_1)^2}\right)\geq (1-q)^2 + q^2.
\end{align*}
But this is a contradiction, because
the left-hand side is less than or equal to $q^2$,
and the right-hand side is strictly bigger than $q^2$.

Finally we prove that if $q>1/2$, then $\lambda_2>J$.
Assuming the contrary, we see
\begin{align*}
\lambda_2\leq J 
&\implies q(q(\mu_0+\mu_1)+\sqrt{\delta})\leq (2q-1)(\mu_0+\mu_1)\\
&\implies q\sqrt{\delta} \leq -(1-q)^2 (\mu_0+\mu_1),
\end{align*}
which is an immediate contradition.

\bibliographystyle{plain}

\def\cprime{$'$}

\end{document}